\pdfoutput=1
\documentclass[nonacm, manuscript, screen, sigplan]{acmart}

\usepackage{stmaryrd}
\usepackage{xspace}
\usepackage{tikz}
\usetikzlibrary{shapes.geometric}
\usetikzlibrary{positioning}

\AtBeginDocument{%
  }

\usepackage{amsmath}
\usepackage{amsthm}
\usepackage{mathtools}
\usepackage{mathpartir}
\usepackage{stmaryrd}
\usepackage{graphicx}
\usepackage[capitalize]{cleveref}

\newcommand{\ttt}[1]{\textnormal{\texttt{#1}}}

\def\grad{\widetilde}

\newcommand{\htrip}[3]{\vdash \{#1\} ~ #2 ~ \{#3\}}
\newcommand{\itrip}[3]{\vdash [#1] ~#2~ [#3]}

\newcommand{\etrip}[3]{\vdash (#1) ~ #2 ~ (#3)}
\newcommand{\getrip}[3]{\gvdash (#1) ~ #2 ~ (#3)}

\newcommand{\Expr}{\mathrm{Expr}}

\newcommand{\Formula}{\textsc{Formula}}
\newcommand{\SatFormula}{\textsc{SatFormula}}
\newcommand{\GFormula}{\ensuremath{\grad{\textsc{F}}\textsc{ormula}}}

\DeclareMathOperator{\fv}{fv}

\renewcommand{\wp}{\operatorname{wp}}
\renewcommand{\sp}{\operatorname{sp}}
\renewcommand{\mod}{\operatorname{mod}}

\newcommand{\iffdef}{\mathrel{\stackrel{\text{def}}{\iff}}}
\newcommand{\gvdash}{\mathrel{\grad{\vdash}}}

\def\kskip{\ttt{skip}}
\def\kif{\ttt{if}}
\def\kthen{\ttt{then}}
\def\kelse{\ttt{else}}
\newcommand{\ifstmt}[3]{\kif\ #1\ \kthen\ #2\ \kelse\ #3}

\makeatletter
\providecommand*{\Dashv}{%
  \mathrel{%
    \mathpalette\@Dashv\vDash
  }%
}
\newcommand*{\@Dashv}[2]{%
  \reflectbox{$\m@th#1#2$}%
}
\makeatother

\newtheorem{theorem}{Theorem}
\newtheorem{lemma}{Lemma}

\theoremstyle{definition}
\newtheorem{definition}{Definition}

\newcommand{\GEL}{\texorpdfstring{\ensuremath{\widetilde{\text{EL}}}}{\~EL}\xspace}
\newcommand{\GIL}{\ensuremath{\widetilde{\text{IL}}}\xspace}

\begin{document}

\title[Gradual Exact Logic]{
    \texorpdfstring{Gradual Exact Logic \\ \LARGE Unifying Hoare Logic and Incorrectness Logic via Gradual Verification}
    {Gradual Exact Logic: Unifying Hoare Logic and Incorrectness Logic via Gradual Verification}}

\author{Conrad Zimmerman}
\email{zimmerman.co@northeastern.edu}
\affiliation{%
  \institution{Northeastern University}
  \city{Boston}
  \state{MA}
  \country{USA}
}

\author{Jenna DiVincenzo}
\email{jennad@purdue.edu}
\affiliation{%
  \institution{Purdue University}
  \city{West Lafayette}
  \state{IN}
  \country{USA}
}


\begin{abstract}
  Previously, gradual verification has been developed using overapproximating logics such as Hoare logic. We show that the static verification component of gradual verification is also connected to underapproximating logics like incorrectness logic.
  To do this, we use a novel definition of gradual verification and a novel gradualization of exact logic \cite{el} which we call \emph{gradual exact logic}. Further, we show that Hoare logic, incorrectness logic, and gradual verification can be defined in terms of gradual exact logic.

  We hope that this connection can be used to develop tools and techniques that apply to both gradual verification and bug-finding.
  For example, we envision that techniques defined in terms of exact logic can be directly applied to verification, bug-finding, and gradual verification, using the principles of gradual typing \cite{agt}.
\end{abstract}





\maketitle

\section{Overview}\label{sec:intro}

Incorrectness logic \cite{il} has been recently developed as a formal basis for ``true bug-finding'' and has been applied in industrial-strength tools \cite{pulsex}. Deductions in this logic prove reachability, which enables bug-finding tools to prove the existence of an invalid state while selectively exploring the possible paths.

At the same time, gradual verification (GV) \cite{vmcai18} addresses the complexity of traditional static verification. Gradually verified programs may contain \emph{imprecise specifications}---logical formulas annotated to indicate that they contain only a partial specification of behavior. A gradual verifier checks the imprecise specifications using static verification where it can and run-time checks (i.e. dynamic verification) elsewhere. These run-time checks can be exercised, e.g. with testing, giving the programmer confidence that their code will not enter a state that violates their partial specifications. Using gradual verification, programmers can incrementally verify a program, incrementally learn verification constructs, and safely guard unverified components.

More recent work \cite{ol,el,gillian} has produced logics and tools that unify over-approximating (OX) logic, often used in verification, and under-approximating (UX) logics including IL. In this paper, we propose another unification of OX and UX logics that we derive using the principles of gradual verification. The resulting characterization of IL demonstrates a previously unexplored connection between GV and IL.

\subsection{Hoare logic}\label{sec:overview-hl}

A triple in Hoare logic (HL) \cite{hoare} is denoted $\{ P \} ~C~ \{ Q \}$. ($P, Q$ refer to logical specifications; $C$ refers to a program statement.) Semantically, the triple is valid if, for every state $\sigma \in P$ (i.e., $P$ is true of $\sigma$), when $\sigma \stackrel{C}{\to} \sigma'$ (i.e., executing $C$ results in $\sigma'$), then $\sigma' \in Q$:
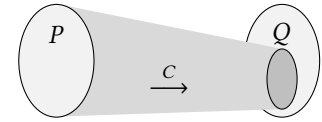
\begin{center}
  \begin{tikzpicture}
    \node[ellipse,
      draw,
      anchor=north,
      minimum width=1cm,
      fill=gray!10,
      minimum height=1.5cm] (Q) at (3, 0) {};
    \fill[fill=gray!30] (0, 0) -- (3, -0.6) -- (3, -1.4) -- (0, -1.5);
    \node[ellipse, draw, fill=gray!50, minimum height=0.8cm, minimum width=0.4cm] at (3, -1) {};
    \node[ellipse,
      draw,
      anchor=north,
      minimum width=1cm,
      fill=gray!10,
      minimum height=1.5cm] (P) at (0, 0) {};
    \node at (0, -0.4) {$P$};
    \node at (3, -0.4) {$Q$};
    \node at (1.5, -1) {$\overset{C}{\longrightarrow}$};
  \end{tikzpicture}
  \end{center}

Thus HL is an \emph{overapproximating} (OX) logic---the postcondition $Q$ overapproximates the states that are reachable from $P$; precisely, $Q \supseteq \{ \sigma' \mid \exists \sigma \in P : \sigma \overset{C}{\to} \sigma' \}$.
HL is a formal foundation for program verification precisely because the postcondition is true in \emph{all} ending states.

\subsection{Incorrectness logic}\label{sec:overview-il}

A triple in incorrectness logic (IL) \cite{il} is denoted $[P] ~C~ [Q]$\footnote{IL triples often include a specification for error states, but we omit error handling to simplify the comparison with other logics.}. Semantically, the triple is valid if, for every $\sigma' \in Q$, there is some $\sigma \in P$ such that $\sigma \overset{C}{\to} \sigma'$:
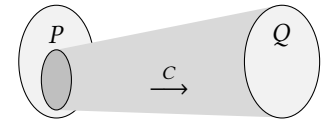
\begin{center}
  \begin{tikzpicture}
    \node[ellipse,
      draw,
      anchor=north,
      minimum width=1cm,
      fill=gray!10,
      minimum height=1.5cm] (P) at (0, 0) {};
    \fill[fill=gray!30] (3, 0) -- (0, -0.6) -- (0, -1.4) -- (3, -1.5);
    \node[ellipse, draw, fill=gray!50, minimum height=0.8cm, minimum width=0.4cm] at (0, -1) {};
    \node[ellipse,
      draw,
      anchor=north,
      minimum width=1cm,
      fill=gray!10,
      minimum height=1.5cm] (Q) at (3, 0) {};
    \node at (0, -0.4) {$P$};
    \node at (3, -0.4) {$Q$};
    \node at (1.5, -1) {$\stackrel{C}{\longrightarrow}$};
  \end{tikzpicture}
  \end{center}
Thus IL is an \emph{underapproximating} (UX) logic---the postcondition $Q$ underapproximates the states that are reachable from $P$; precisely, $Q \subseteq \{ \sigma' \mid \exists \sigma \in P : \sigma \overset{C}{\to} \sigma' \}$. Interpreting $Q$ as a specification of a bug, IL is a logic for finding bugs since a valid triple indicates that the bug is reachable.

\subsection{Gradual verification}\label{sec:overview-gv}

Gradual verification (GV) \citep{vmcai18} reduces the burden of static verification by allowing incomplete (\emph{imprecise}) specifications. A gradual verifier may make optimistic assumptions when verifying imprecisely-specified code. Typically, the final program is elaborated to check these assumptions at run-time. However, in this work we focus solely on the static verification component of GV so that we can compare its logic with HL and IL.


We denote imprecise triples by $\{ \ttt{?} \wedge P \} ~C~ \{ Q \}$. Intuitively, this triple is valid if there is some $P' \Rightarrow P$ such that $\{ P' \} ~C~ \{ Q \}$ is valid in HL. One can think of \ttt{?} as representing the additional assumptions introduced by $P'$.

For example, the following imprecise triple is valid:
$$\{ \textcolor{purple}{\ttt{?}} \wedge \top \} ~x \coloneq x + 1~ \{ x > 0 \}$$
This follows from the validity of
$$\{ \textcolor{purple}{x \ge 0} \wedge \top \} ~x \coloneq x + 1~ \{ x > 0 \}.$$
That is, the postcondition is ensured when assuming $x \ge 0$.

However, the assumptions must be plausible; formally, $P'$ (and thus $P$) must be satisfiable (i.e., $P' \not\equiv \bot$). Otherwise, all imprecise triples would be vacuously valid by taking $P' \equiv \bot$.
With this in mind, GV can be stated as a \emph{reachability} problem---semantically, $\{ \ttt{?} \wedge P \} ~C~ \{ Q \}$ is valid if there exists some states $\sigma$ and $\sigma'$ such that $\sigma \in P$, $\sigma \overset{C}{\to} \sigma'$, and $\sigma' \in Q$:
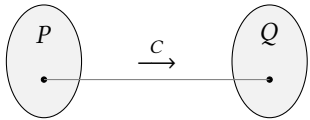
\begin{center}
  \begin{tikzpicture}
    \node[ellipse,
      draw,
      anchor=north,
      minimum width=1cm,
      fill=gray!10,
      minimum height=1.5cm] (Q) at (3, 0) {};
    \node[ellipse,
      draw,
      anchor=north,
      minimum width=1cm,
      fill=gray!10,
      minimum height=1.5cm] (P) at (0, 0) {};
    \filldraw (0, -1) circle (1pt);
    \filldraw (3, -1) circle (1pt);
    \node at (0, -0.4) {$P$};
    \node at (3, -0.4) {$Q$};
    \draw[color=gray] (0, -1) -- node[color=black, above] {$\stackrel{C}{\longrightarrow}$} (3, -1);
  \end{tikzpicture}
\end{center}
By contrapositive, the triple is invalid (i.e., static verification will error) if it is never possible for $C$ to ensure $Q$, given $P$.

By comparing this diagram with those in \S\ref{sec:overview-hl} and \S\ref{sec:overview-il}, we can deduce that HL \emph{and} IL triples are valid imprecise triples, except for vacuous cases where $P \equiv \bot$ in HL or $Q \equiv \bot$ in IL.

Our semantic definition of validity is equivalent to the previous definition which uses Hoare triples: Let $P'$ be a formula that represents $\sigma$ as specifically as possible, then (intuitively) $\{ P' \} ~C~ \{ Q \}$ is valid since $\sigma \overset{C}{\to} \sigma'$ and $\sigma' \in Q$.

But, we can also define valid imprecise triples in terms of IL---$\{ \ttt{?} \wedge P \} ~C~ \{ Q \}$ is valid if $[P] ~C~ [Q']$ is valid for some $Q' \Rightarrow Q$. In \S\ref{sec:gv-hl-gel-il} we will prove these definitions equivalent.

\section{Formal foundations}

We will now sketch our formal definitions and results. See the appendices for the full statements and proofs.
\begin{enumerate}
  \item We define \emph{gradual exact logic}---a consistent lifting \citep{agt} of exact logic \cite{el}.
  \item We show that HL, IL, and GV can be characterized by gradual exact logic.
  \item We show that GV contains OX and UX deductions.
  \item We show that, for imprecise specifications, GV can be equivalently defined using OX or UX logics.
\end{enumerate}

\subsection{Exact logic}

Exact logic (EL) \cite{el} is the intersection of HL and IL: an EL triple $(P)~C~(Q)$ is valid if $\{ P \}~C~\{Q\}$ and $[P]~C~[Q]$ are both valid. Deductions are thus \emph{exact}---they can neither under- or overapproximate behavior (see \cref{app:el} for rules).

\subsection{Gradual exact logic}

We further define \emph{gradual exact logic} (\GEL) as a consistent lifting \citep{agt} of EL.

First, some definitions: $\Formula$ denotes all formulas in FOL with arithmetic; we call these \emph{precise}. $\SatFormula$ denotes all satisfiable formulas.
An \emph{imprecise formula} is of the form $\ttt{?} \wedge P$ where $P \in \Formula$.
A \emph{gradual formula} $\grad{P} \in \GFormula$ can be either precise or imprecise.

The \emph{concretization} $\gamma : \GFormula \to \mathcal{P}(\Formula)$ interprets a gradual formulas as sets of precise formulas:
$$
  \gamma(\ttt{?} \wedge P) = \{ P' \in \SatFormula \mid P' \Rightarrow P \}, \quad \gamma(P) = \{ P \}$$
Let $\vdash (P) ~C~ (Q)$ denote a valid EL triple. Deductions in \GEL, denoted $\gvdash (\grad{P}) ~C~ (\grad{Q})$, are defined as a consistent lifting of EL deductions (\cref{app:el-gradual}):
$$\gvdash (\grad{P})~C~(\grad{Q}) \iffdef \vdash (P)~C~(Q) ~\text{for some $P \in \gamma(\grad{P})$, $Q \in \gamma(\grad{Q})$}$$

\subsection{Strongest postconditions}

$\sp(P, C)$ denotes the strongest (WRT $\Rightarrow$) $Q$ for which $\{ P \} ~C~ \{ Q \}$ is valid (calculated as usual; see \Cref{app:sp}). Strongest postconditions are related to HL, IL, and EL as follows:
\begin{align*}
  \vdash \{ P \} ~C~ \{ Q \} &\iff \sp(P, C) \Rightarrow Q &\text{\Cref{thm:hoare-iff-sp}}\\
  \vdash [P] ~C~ [Q] &\iff \sp(P, C) \Leftarrow Q &\text{\Cref{thm:il-iff-sp}}\\
  \vdash (P) ~C~ (Q) &\iff \sp(P, C) \equiv Q &\text{\Cref{thm:el-iff-sp}}
\end{align*}



\subsection{HL and IL via \GEL}\label{sec:charact-hl-il}

We can characterize valid HL triples as \GEL triples where the postcondition is made imprecise. Assuming that $P \not\equiv \bot$ and $C$ terminates (in either of these cases $\{P\}~C~\{Q\}$ is vacuous), we have $P, Q, \sp(P, C) \in \SatFormula$ and thus (\cref{thm:hl-as-gel})
\begin{align*}
  \vdash \{P\}~C~\{Q\} &\iff \sp(P, C) \Rightarrow Q \\
    &\iff \sp(P, C) \in \gamma(\ttt{?} \wedge Q) \\
    &\iff \getrip{P}{C}{\ttt{?} \wedge Q}.
\end{align*}

Likewise, we can characterize valid IL triples as \GEL triples where the precondition is made imprecise. Weakest preconditions are not always defined for IL \citep{il}, however, we can reuse weakest preconditions for HL to witness the necessary formula (see \cref{app:fixpoint}). Assuming $Q \not\equiv \bot$ (otherwise $[P]~C~[Q]$ is vacuous), we have (\cref{thm:il-as-gel})
\begin{align*}
  \itrip{P}{C}{Q} &\iff Q \Rightarrow \sp(P, C) \\
    &\iff Q \equiv \sp(\wp(Q, C) \wedge P, C) \\
    &\iff \etrip{P \wedge \wp(Q, C)}{C}{Q} \\
    &\iff \getrip{\ttt{?} \wedge P}{C}{Q}.
\end{align*}

\subsection{GV via HL, \GEL, and IL}\label{sec:gv-hl-gel-il}

We can now give a precise definition of GV in terms of HL, as sketched in \S\ref{sec:overview-gv}\footnote{Here we define GV more generally for all \emph{gradual} preconditions, whereas in \S\ref{sec:overview-gv} we defined it for \emph{imprecise} preconditions.}. We denote valid GV triples as $\gvdash \{ \grad{P} \} ~C~ \{ Q \}$.
$$\gvdash \{ \grad{P} \} ~C~ \{ Q \} \iffdef \vdash \{ P \}~C~\{Q\} ~\text{for some $P \in \gamma(\grad{P})$}$$
Using this definition and applying the characterization of HL from \S\ref{sec:charact-hl-il} to characterize GV in terms of EL:
$$\gvdash \{ \grad{P} \} ~C~ \{ Q \} \iff \gvdash (\grad{P})~C~(\ttt{?} \wedge Q)$$

For sake of comparison, we can define \GIL as a lifting of IL, the same way we have lifted HL to GV.
$$\gvdash [P]~C~[\grad{Q}] \iffdef \vdash [P]~C~[Q] ~\text{for some $Q \in \gamma(\grad{Q})$}$$
But, in the case of imprecision this is equivalent to GV. We can see this using the characterization of IL given in \S\ref{sec:charact-hl-il}:
\begin{align*}
  \gvdash [P]~C~[\ttt{?} \wedge Q] &\iff \getrip{\ttt{?} \wedge P}{C}{\ttt{?} \wedge Q} \\
    &\iff \gvdash \{ \ttt{?} \wedge P \} ~C~ \{ Q \}.
\end{align*}

Note: GV and \GIL differ on precise formulas---$\gvdash [P] ~C~ [Q]$ is \emph{not} equivalent to $\gvdash \{ P \} ~C~ \{ Q \}$. Also, we do not gradualize postconditions in HL or preconditions in IL because we can arbitrarily weaken these specifications already.

While verification of precise formulas is a key aspect of GV, this demonstrates that verification of imprecise formulas can be accomplished using IL, and moreover that GV, when verifying imprecisely-specified code, is proving an IL deduction.
Finally, we can make precise our claim from \S\ref{sec:overview-gv} that both OX and UX deductions are valid in (the imprecise fragment of) GV. Assuming $P \not\equiv \bot$, we have (immediate from definition of GV)
$$\htrip{P}{C}{Q} \implies \gvdash \{ \ttt{?} \wedge P \} ~C~ \{ Q \}.$$
Also, assuming $Q \not\equiv \bot$, we have
\begin{align*}
  \itrip{P}{C}{Q} &\implies \getrip{\ttt{?} \wedge P}{C}{Q} \\
    &\implies \getrip{\ttt{?} \wedge P}{C}{\ttt{?} \wedge Q} \\
    &\implies \gvdash \{ \ttt{?} \wedge P \} ~C~ \{ Q \}.
\end{align*}
Thus GV (and \GEL) represents the union of $\text{HL}$ and $\text{IL}$, while EL represents the intersection.

\section{Applications}\label{sec:applications}

GV and IL verifiers in practice operate quite similarly; for example, compare the core \emph{consume} operation of \citet{popl2024} for a gradual verifier and of \citet{gillian} for an IL verifier. Both types of verifiers make assumptions, including pruning paths, when establishing a postcondition. This work formalizes the connection between these methods of verification; in particular, gradual verification of imprecisely-specified code is equivalent to IL deductions. We hope that this will allow verifiers that already incorporate both OX and UX logics, for example \citet{gillian}, to be easily extended to GV.

Similarly, we hope our approach can be used as a framework for unifying techniques across static verification, GV, and bug-finding. For example, bi-abduction has been developed in OX logics \cite{biabduction}, applied to UX verification \cite{gillian}, and is related to GV \cite{jenna}. We expect that techniques like this could be developed in the context of an exact logic, and then their applications to OX, UX, and GV logics could be derived using AGT-style techniques \cite{agt}.

\section{Caveats and future work}\label{sec:caveats}

While we have proven the results described, we have done so only for a very restrictive language, and thus our results should be considered preliminary. In particular, we do not consider heaps, method calls, or loops. We expect our results extend to these constructs, but showing this will require significant work.

Our definition of GV differs from previous definitions \citep{vmcai18,oopsla20,popl2024}. Our novel definition more clearly demonstrates the connection with IL, and we believe it captures the essence of GV. However, this definition only considers static verification, and thus does not consider run-time assertions. We expect this could be added, and we have hopes that we may be able to model run-time assertions using an evidence-based calculus similar to that of gradual typing \citet{agt}.
In addition, further work is necessary to elucidate how the ``gradual guarantees'' \cite{agt} affect the relation of GV to IL. In particular, the static gradual guarantee seems to prohibit arbitrarily dropping paths, which IL can do.

Finally, our work is (to our knowledge) the first to explore a gradualization of exact logic. It remains to be seen whether this is useful its own right. For example, if library developers write exact specifications, a gradual exact logic could be used to aid development of these specifications, similar to how GV aids OX verification. But significantly more work would be necessary for this.

\section{Conclusion}

We have demonstrated the similarities and differences of GV and IL. Specifically, the relation between IL deductions and \GEL deductions with imprecise preconditions shows that the notion of assumptions used in GV is equivalent to the consequence rule (and path pruning) in IL. We also have shown that GV with imprecise specifications is equivalent to IL. Furthermore, we have defined gradual exact logic and used this to formally compare IL, HL, and GV. While much work remains before it is widely applicable, we hope that this framework can be used to develop techniques that uniformly target all three methods of verification.

\begin{acks}
  We thank Devin Singh for his contributions to this work.
\end{acks}

\clearpage

\bibliographystyle{ACM-Reference-Format}
\bibliography{references}


\begin{thebibliography}{12}


\ifx \showCODEN    \undefined \def \showCODEN     #1{\unskip}     \fi
\ifx \showDOI      \undefined \def \showDOI       #1{#1}\fi
\ifx \showISBNx    \undefined \def \showISBNx     #1{\unskip}     \fi
\ifx \showISBNxiii \undefined \def \showISBNxiii  #1{\unskip}     \fi
\ifx \showISSN     \undefined \def \showISSN      #1{\unskip}     \fi
\ifx \showLCCN     \undefined \def \showLCCN      #1{\unskip}     \fi
\ifx \shownote     \undefined \def \shownote      #1{#1}          \fi
\ifx \showarticletitle \undefined \def \showarticletitle #1{#1}   \fi
\ifx \showURL      \undefined \def \showURL       {\relax}        \fi
\providecommand\bibfield[2]{#2}
\providecommand\bibinfo[2]{#2}
\providecommand\natexlab[1]{#1}
\providecommand\showeprint[2][]{arXiv:#2}

\bibitem[Bader et~al\mbox{.}(2018)]%
        {vmcai18}
\bibfield{author}{\bibinfo{person}{Johannes Bader}, \bibinfo{person}{Jonathan Aldrich}, {and} \bibinfo{person}{{\'{E}}ric Tanter}.} \bibinfo{year}{2018}\natexlab{}.
\newblock \showarticletitle{Gradual Program Verification}. In \bibinfo{booktitle}{\emph{Verification, Model Checking, and Abstract Interpretation - 19th International Conference, {VMCAI} 2018, Los Angeles, CA, USA, January 7-9, 2018, Proceedings}} \emph{(\bibinfo{series}{Lecture Notes in Computer Science}, Vol.~\bibinfo{volume}{10747})}, \bibfield{editor}{\bibinfo{person}{Isil Dillig} {and} \bibinfo{person}{Jens Palsberg}} (Eds.). \bibinfo{publisher}{Springer}, \bibinfo{pages}{25--46}.
\newblock
\urldef\tempurl%
\url{https://doi.org/10.1007/978-3-319-73721-8\_2}
\showDOI{\tempurl}


\bibitem[Calcagno et~al\mbox{.}(2011)]%
        {biabduction}
\bibfield{author}{\bibinfo{person}{Cristiano Calcagno}, \bibinfo{person}{Dino Distefano}, \bibinfo{person}{Peter~W. O’Hearn}, {and} \bibinfo{person}{Hongseok Yang}.} \bibinfo{year}{2011}\natexlab{}.
\newblock \showarticletitle{Compositional Shape Analysis by Means of Bi-Abduction}.
\newblock \bibinfo{journal}{\emph{J. ACM}} \bibinfo{volume}{58}, \bibinfo{number}{6}, Article \bibinfo{articleno}{26} (\bibinfo{date}{Dec.} \bibinfo{year}{2011}), \bibinfo{numpages}{66}~pages.
\newblock
\showISSN{0004-5411}
\urldef\tempurl%
\url{https://doi.org/10.1145/2049697.2049700}
\showDOI{\tempurl}


\bibitem[DiVincenzo(2023)]%
        {jenna}
\bibfield{author}{\bibinfo{person}{Jenna~Wise DiVincenzo}.} \bibinfo{year}{2023}\natexlab{}.
\newblock \emph{\bibinfo{title}{Gradual Verification of Recursive Heap Data Structures}}.
\newblock \bibinfo{thesistype}{Ph.\,D. Dissertation}. \bibinfo{school}{Carnegie Mellon University}.
\newblock


\bibitem[Garcia et~al\mbox{.}(2016)]%
        {agt}
\bibfield{author}{\bibinfo{person}{Ronald Garcia}, \bibinfo{person}{Alison~M. Clark}, {and} \bibinfo{person}{{\'{E}}ric Tanter}.} \bibinfo{year}{2016}\natexlab{}.
\newblock \showarticletitle{Abstracting gradual typing}. In \bibinfo{booktitle}{\emph{Proceedings of the 43rd Annual {ACM} {SIGPLAN-SIGACT} Symposium on Principles of Programming Languages, {POPL} 2016, St. Petersburg, FL, USA, January 20 - 22, 2016}}, \bibfield{editor}{\bibinfo{person}{Rastislav Bod{\'{\i}}k} {and} \bibinfo{person}{Rupak Majumdar}} (Eds.). \bibinfo{publisher}{{ACM}}, \bibinfo{pages}{429--442}.
\newblock
\urldef\tempurl%
\url{https://doi.org/10.1145/2837614.2837670}
\showDOI{\tempurl}


\bibitem[Hoare(1969)]%
        {hoare}
\bibfield{author}{\bibinfo{person}{C.~A.~R. Hoare}.} \bibinfo{year}{1969}\natexlab{}.
\newblock \showarticletitle{An Axiomatic Basis for Computer Programming}.
\newblock \bibinfo{journal}{\emph{Commun. {ACM}}} \bibinfo{volume}{12}, \bibinfo{number}{10} (\bibinfo{year}{1969}), \bibinfo{pages}{576--580}.
\newblock
\urldef\tempurl%
\url{https://doi.org/10.1145/363235.363259}
\showDOI{\tempurl}


\bibitem[Le et~al\mbox{.}(2022)]%
        {pulsex}
\bibfield{author}{\bibinfo{person}{Quang~Loc Le}, \bibinfo{person}{Azalea Raad}, \bibinfo{person}{Jules Villard}, \bibinfo{person}{Josh Berdine}, \bibinfo{person}{Derek Dreyer}, {and} \bibinfo{person}{Peter~W. O'Hearn}.} \bibinfo{year}{2022}\natexlab{}.
\newblock \showarticletitle{Finding real bugs in big programs with incorrectness logic}.
\newblock \bibinfo{journal}{\emph{Proc. {ACM} Program. Lang.}} \bibinfo{volume}{6}, \bibinfo{number}{{OOPSLA1}} (\bibinfo{year}{2022}), \bibinfo{pages}{1--27}.
\newblock
\urldef\tempurl%
\url{https://doi.org/10.1145/3527325}
\showDOI{\tempurl}


\bibitem[L{\"{o}}{\"{o}}w et~al\mbox{.}(2024)]%
        {gillian}
\bibfield{author}{\bibinfo{person}{Andreas L{\"{o}}{\"{o}}w}, \bibinfo{person}{Daniele Nantes{-}Sobrinho}, \bibinfo{person}{Sacha{-}{\'{E}}lie Ayoun}, \bibinfo{person}{Caroline Cronj{\"{a}}ger}, \bibinfo{person}{Petar Maksimovic}, {and} \bibinfo{person}{Philippa Gardner}.} \bibinfo{year}{2024}\natexlab{}.
\newblock \showarticletitle{Compositional Symbolic Execution for Correctness and Incorrectness Reasoning}. In \bibinfo{booktitle}{\emph{38th European Conference on Object-Oriented Programming, {ECOOP} 2024, September 16-20, 2024, Vienna, Austria}} \emph{(\bibinfo{series}{LIPIcs}, Vol.~\bibinfo{volume}{313})}, \bibfield{editor}{\bibinfo{person}{Jonathan Aldrich} {and} \bibinfo{person}{Guido Salvaneschi}} (Eds.). \bibinfo{publisher}{Schloss Dagstuhl - Leibniz-Zentrum f{\"{u}}r Informatik}, \bibinfo{pages}{25:1--25:28}.
\newblock
\urldef\tempurl%
\url{https://doi.org/10.4230/LIPICS.ECOOP.2024.25}
\showDOI{\tempurl}


\bibitem[Maksimovic et~al\mbox{.}(2023)]%
        {el}
\bibfield{author}{\bibinfo{person}{Petar Maksimovic}, \bibinfo{person}{Caroline Cronj{\"{a}}ger}, \bibinfo{person}{Andreas L{\"{o}}{\"{o}}w}, \bibinfo{person}{Julian Sutherland}, {and} \bibinfo{person}{Philippa Gardner}.} \bibinfo{year}{2023}\natexlab{}.
\newblock \showarticletitle{Exact Separation Logic: Towards Bridging the Gap Between Verification and Bug-Finding}. In \bibinfo{booktitle}{\emph{37th European Conference on Object-Oriented Programming, {ECOOP} 2023, July 17-21, 2023, Seattle, Washington, United States}} \emph{(\bibinfo{series}{LIPIcs}, Vol.~\bibinfo{volume}{263})}, \bibfield{editor}{\bibinfo{person}{Karim Ali} {and} \bibinfo{person}{Guido Salvaneschi}} (Eds.). \bibinfo{publisher}{Schloss Dagstuhl - Leibniz-Zentrum f{\"{u}}r Informatik}, \bibinfo{pages}{19:1--19:27}.
\newblock
\urldef\tempurl%
\url{https://doi.org/10.4230/LIPICS.ECOOP.2023.19}
\showDOI{\tempurl}


\bibitem[O'Hearn(2020)]%
        {il}
\bibfield{author}{\bibinfo{person}{Peter~W. O'Hearn}.} \bibinfo{year}{2020}\natexlab{}.
\newblock \showarticletitle{Incorrectness logic}.
\newblock \bibinfo{journal}{\emph{Proc. {ACM} Program. Lang.}} \bibinfo{volume}{4}, \bibinfo{number}{{POPL}} (\bibinfo{year}{2020}), \bibinfo{pages}{10:1--10:32}.
\newblock
\urldef\tempurl%
\url{https://doi.org/10.1145/3371078}
\showDOI{\tempurl}


\bibitem[Wise et~al\mbox{.}(2020)]%
        {oopsla20}
\bibfield{author}{\bibinfo{person}{Jenna Wise}, \bibinfo{person}{Johannes Bader}, \bibinfo{person}{Cameron Wong}, \bibinfo{person}{Jonathan Aldrich}, \bibinfo{person}{{\'{E}}ric Tanter}, {and} \bibinfo{person}{Joshua Sunshine}.} \bibinfo{year}{2020}\natexlab{}.
\newblock \showarticletitle{Gradual verification of recursive heap data structures}.
\newblock \bibinfo{journal}{\emph{Proc. {ACM} Program. Lang.}} \bibinfo{volume}{4}, \bibinfo{number}{{OOPSLA}} (\bibinfo{year}{2020}), \bibinfo{pages}{228:1--228:28}.
\newblock
\urldef\tempurl%
\url{https://doi.org/10.1145/3428296}
\showDOI{\tempurl}


\bibitem[Zilberstein et~al\mbox{.}(2023)]%
        {ol}
\bibfield{author}{\bibinfo{person}{Noam Zilberstein}, \bibinfo{person}{Derek Dreyer}, {and} \bibinfo{person}{Alexandra Silva}.} \bibinfo{year}{2023}\natexlab{}.
\newblock \showarticletitle{Outcome Logic: {A} Unifying Foundation for Correctness and Incorrectness Reasoning}.
\newblock \bibinfo{journal}{\emph{Proc. {ACM} Program. Lang.}} \bibinfo{volume}{7}, \bibinfo{number}{{OOPSLA1}} (\bibinfo{year}{2023}), \bibinfo{pages}{522--550}.
\newblock
\urldef\tempurl%
\url{https://doi.org/10.1145/3586045}
\showDOI{\tempurl}


\bibitem[Zimmerman et~al\mbox{.}(2024)]%
        {popl2024}
\bibfield{author}{\bibinfo{person}{Conrad Zimmerman}, \bibinfo{person}{Jenna DiVincenzo}, {and} \bibinfo{person}{Jonathan Aldrich}.} \bibinfo{year}{2024}\natexlab{}.
\newblock \showarticletitle{Sound Gradual Verification with Symbolic Execution}.
\newblock \bibinfo{journal}{\emph{Proc. {ACM} Program. Lang.}} \bibinfo{volume}{8}, \bibinfo{number}{{POPL}} (\bibinfo{year}{2024}), \bibinfo{pages}{2547--2576}.
\newblock
\urldef\tempurl%
\url{https://doi.org/10.1145/3632927}
\showDOI{\tempurl}


\end{thebibliography}

\clearpage

\appendix

\setlength{\abovedisplayskip}{2pt}
\setlength{\belowdisplayskip}{2pt}
\setlength{\abovedisplayshortskip}{2pt}
\setlength{\belowdisplayshortskip}{2pt}
\def \MathparLineskip {\lineskip=2pt}

\section{Grammar}\label{app:grammar}

We define a basic language that includes mutable variables and conditionals. Note that we do not include \ttt{while} loops or functions; these constructs are left for future work.

\begin{align*}
  \Expr \ni E ~{\Coloneqq}~ &n \mid \bot \mid \top \mid x \mid E \vee E \mid E \wedge E \mid \\
    &E = E \mid E < E \mid \neg E \\
  \mathrm{Cmd} \ni C ~{\Coloneqq}~ &\kskip \mid x \coloneq E \mid \ifstmt{E}{C}{C} \mid C; C \\
  \mathrm{Asrt} \ni P, Q, R ~{\Coloneqq}~ &E \mid \neg P \mid P \wedge P \mid P \vee P \mid \exists x\, P
\end{align*}
where $n \in \mathbb{Z}$ and $x$ a variable name.

We assume that programs are typed correctly; for example, expressions $E$ in $\ifstmt{E}{C_1}{C_2}$ will always evaluate to a boolean value.

\begin{definition}[Implication]
  $P \Rightarrow Q$ if all states that satisfy $P$ also satisfy $Q$.

  Note: we do not formalize states or the semantics of assertions; we assume that our logic is close enough to first-order logic and thus use FOL deductions to reason about assertions.

  Note: We use ${\cdot} \Rightarrow {\cdot}$ to denote implication between formulas in first-order logic with booleans and arithmetic. ${\cdot} \implies {\cdot}$ denotes ``if-then'' in our metatheory.
\end{definition}

\begin{definition}[Equivalence of propositions]
  $P \equiv Q$ denotes that $P$ and $Q$ are logically equivalent; that is,
  $$P \equiv Q \iffdef P \Rightarrow Q ~\text{and}~ Q \Rightarrow P.$$
  Throughout this paper, we assume that identity of propositions coincides with equivalence; that is, when we write an individual proposition, technically we are denoting the equivalence class that contains that proposition. For example, $\{ \top, \bot \} = \{ 1 = 1, 1 = 2 \}$.
\end{definition}

\begin{definition}[Replacement]
  $P[x/E]$ denotes $P$ with all free occurrences of $x$ replaced by $E$.
\end{definition}

\begin{definition}
  $\fv(P)$ denotes the free variables in the assertion $P$.
\end{definition}

\begin{definition}
  $\fv(C)$ denotes all variables referenced or assigned in $C$.
\end{definition}

\begin{definition}
  $\mod(C)$ denotes all variables assigned in $C$. Explicitly,
  \begin{align*}
    \mod(\kskip) &\coloneq \emptyset \\
    \mod(x \coloneq E) &\coloneq x \\
    \mod(C_1; C_2) &\coloneq \mod(C_1) \cup \mod(C_2) \\
    \mod(\ifstmt{E}{C_1}{C_2}) &\coloneq \mod(C_1) \cup \mod(C_2)
  \end{align*}
\end{definition}

\section{Predicate transformers}\label{app:transformers}

\subsection{Weakest preconditions}\label{app:wp}

\begin{definition}
  $P$ is the \emph{weakest precondition} for a statement $C$ and postcondition $Q$, denoted $\wp(C, Q)$, if it is the weakest predicate (WRT $\Rightarrow$) that ensures that if a state satisfies $P$, then after executing $C$, $Q$ holds.
\end{definition}

The following explicit calculations correspond to the previous definition (proved in \cref{thm:wp-iff-hoare}):
\begin{align*}
  &\wp(\kskip, Q) = Q \\
  &\wp(x \coloneq E, Q) = Q[x/E] \\
  &\wp(C_1; C_2, Q) = \wp(C_1, \wp(C_2, Q)) \\
  &\wp(\ifstmt{E}{C_1}{C_2}, Q) = \\
  &\quad(\wp(C_1, Q) \wedge E) \vee (\wp(C_2, Q) \wedge \neg E)
\end{align*}

\begin{lemma}[Stronger postcondition]\label{lem:wp-strengthen-post}
  If $Q \Rightarrow Q'$ then \\ $\wp(C, Q) \Rightarrow \wp(C, Q')$.
\end{lemma}
\begin{proof}
  By induction on $C$:
  \begin{description}
    \item[$\kskip$:]
      $$\wp(\kskip, Q) \equiv Q \Rightarrow Q' \equiv \wp(\kskip, Q')$$
    \item[$x \coloneq E$:]
      $$\wp(x \coloneq E, Q) \equiv Q[E/x] \Rightarrow Q'[E/x] \equiv \wp(x \coloneq E, Q')$$
    \item[$C_1; C_2$:]\leavevmode
      \begin{description}
        \item[(1)] $\wp(C_2, Q) \Rightarrow \wp(C_2, Q')$ by induction
        \item[(2)] $\wp(C_1, \wp(C_2, Q)) \Rightarrow \wp(C_1, \wp(C_2, Q'))$ by induction using (1)
        \item[(3)] $\wp(C_1; C_2, Q) \Rightarrow \wp(C_1; C_2, Q')$ by (2) and definition of $\wp$
      \end{description}

    \item[$\ifstmt{E}{C_1}{C_2}$:]~
      \begin{description}
        \item[(1)] $\wp(C_1, Q) \Rightarrow \wp(C_1, Q')$ by induction
        \item[(2)] $\wp(C_2, Q) \Rightarrow \wp(C_2, Q')$ by induction
        \item[(3)] By (1) and (2)
          \begin{align*}
            &(\wp(C_1, Q) \wedge E) \vee (\wp(C_2, Q) \wedge \neg E) \\
            &\quad \Rightarrow (\wp(C_1, Q') \wedge E) \vee (\wp(C_2, Q') \wedge \neg E)
          \end{align*}
        \item[(4)] By (3) and definition of $\wp$
          \begin{align*}
            &\wp(\ifstmt{E}{C_1}{C_2}, Q) \\
            &\quad \Rightarrow \wp(\ifstmt{E}{C_1}{C_2}, Q') \qedhere
          \end{align*}
      \end{description}
  \end{description}
\end{proof}

\subsection{Strongest postconditions}\label{app:sp}

\begin{definition}
  $Q$ is the \emph{strongest postcondition} for precondition $P$ and statement $C$, denoted $\sp(P, C)$, if it is the strongest predicate (WRT $\Rightarrow$) that ensures that if a state satisfies $P$, then after executing $C$, $Q$ holds.
\end{definition}

The following explicit calculations correspond to the previous definition (proved in \cref{thm:hoare-iff-sp}):
\begin{align*}
  &\sp(P, \kskip) = P \\
  &\sp(P, x \coloneq E) = \exists v (x = E[x/v] \wedge P[x/v]) \\
    &\quad \text{where $v \notin \fv(P)$}\\
  &\sp(P, C_1; C_2) = \sp(\sp(P, C_1), C_2) \\
  &\sp(P, \ifstmt{E}{C_1}{C_2}) =\\
  &\quad \sp(P \wedge E, C_1) \vee \sp(P \wedge \neg E, C_2)
\end{align*}

\begin{lemma}[Stronger precondition]\label{lem:sp-strengthen}
  If $P \Rightarrow P'$ then \\ $\sp(P, C) \Rightarrow \sp(P', C)$.
\end{lemma}
\begin{proof}
  By induction on $C$:

  \begin{description}
    \item[$\kskip$:]
      $\sp(P, \kskip) \equiv P \Rightarrow P' \equiv \sp(P', \kskip)$.
    \item[$x \coloneq E$:] Note that $P[x/v] \Rightarrow P'[x/v]$ by logic. Then:
      \begin{align*}
        \sp(P, x \coloneq E) &\equiv \exists v (x = E[x/v] \wedge P[x/v]) &\text{defn $\sp$} \\
          &\Rightarrow \exists v (x = E[x/v] \wedge P'[x/v]) &\text{logic} \\
          &\equiv \sp(P', x \coloneq E) &\text{defn $\sp$}
      \end{align*}
    \item[$C_1; C_2$:]\leavevmode
      \begin{description}
        \item[(1)] $\sp(P, C_1) \Rightarrow \sp(P', C_1)$ by induction
        \item[(2)] $\sp(\sp(P, C_1), C_2) \Rightarrow \sp(\sp(P', C_1), C_2)$ by induction using (1)
        \item[(3)] By (2) and definition of $\sp$,
          \begin{align*}
            \sp(P, C_1; C_2) &\equiv \sp(\sp(P, C_1), C_2) \\
              &\Rightarrow \sp(\sp(P', C_1), C_2) \\
              &\equiv \sp(P', C_1; C_2)
          \end{align*}
      \end{description}
    \item[$\ifstmt{E}{C_1}{C_2}$:]\leavevmode
      \begin{description}
        \item[(1)] $P \wedge E \Rightarrow P' \wedge E$ by logic
        \item[(2)] $\sp(P \wedge E, C_1) \Rightarrow \sp(P' \wedge E, C_1)$ by induction using (1)
        \item[(3)] $P \wedge \neg E \Rightarrow P' \wedge \neg E$ by logic
        \item[(4)] $\sp(P \wedge \neg E, C_2) \Rightarrow \sp(P' \wedge \neg E, C_2)$ by induction using (3)
        \item[(5)]
          \begin{align*}
            &\sp(P, \ifstmt{E}{C_1}{C_2}) \\
              &\quad\equiv \sp(P \wedge E, C_1) \vee \sp(P \wedge \neg E, C_2) &\text{defn $\sp$} \\
              &\quad\Rightarrow \sp(P' \wedge E, C_1) \vee \sp(P' \wedge \neg E, C_2) &\text{(2), (4), logic} \\
              &\quad\equiv \sp(P', \ifstmt{E}{C_1}{C_2}) &\text{defn $\sp$} &\qedhere
          \end{align*}
      \end{description}
  \end{description}
\end{proof}

\begin{lemma}\label{lem:sp-false}
  $\sp(\bot, C) \equiv \bot$
\end{lemma}
\begin{proof}
  By induction on $C$:
  \begin{description}
    \item[$\kskip$:]
      $\sp(\bot, \kskip) \equiv \bot$ by definition.
    \item[$x \coloneq E$:]
      $\sp(\bot, x \coloneq E) \equiv \exists v (x = E[x/v] \wedge \bot) \equiv \bot$ by logic.
    \item[$\ifstmt{E}{C_1}{C_2}$:]
      \begin{align*}
        &\sp(\bot, \ifstmt{E}{C_1}{C_2}) \\
          &\quad\equiv \sp(\bot \wedge E, C_1) \vee \sp(\bot \wedge \neg E, C_2) &\text{defn $\sp$}\\
          &\quad\equiv \sp(\bot, C_1) \vee \sp(\bot, C_2) &\text{logic}\\
          &\quad\equiv \bot \vee \bot &\text{induction}\\
          &\quad\equiv \bot &\text{logic}
      \end{align*}
    \item[$C_1; C_2$:]
      \begin{align*}
        \sp(\bot, C_1; C_2) &\equiv \sp(\sp(\bot, C_1), C_2) &\text{defn $\sp$} \\
          &\equiv \sp(\bot, C_2) &\text{induction} \\
          &\equiv \bot &\text{induction} &\qedhere
      \end{align*}
  \end{description}
\end{proof}

\begin{lemma}\label{lem:sp-is-false}
  $\sp(P, C) \equiv \bot \implies P \equiv \bot$

  Note: for a more expressive language, we would also need termination for this to hold.
\end{lemma}
\begin{proof}
  By induction on $C$:
  \begin{description}
    \item[$\kskip$:]
      $\bot \equiv \sp(P, \kskip) \equiv P$ by definition.
    \item[$x \coloneq E$:]\leavevmode
      \begin{description}
        \item[(1)] $\bot \equiv \sp(P, x \coloneq E)$ by assumption
        \item[(2)] $\sp(P, x \coloneq E) \equiv \exists v (x = E[x/v] \wedge P[x/v])$
        \item[(3)] $\top \equiv \forall v (x \ne E[x/v] \vee \neg P[x/v])$ by logic using (1) and (2)
        \item[(4)] $\top \equiv \forall v \neg P[x/v]$ by logic using (3)
        
          We prove this using a semantic argument (assuming standard Kripke semantics for FOL):

          Let $\mathcal{M}$ be a model and $\mathsf{y} \in \lvert \mathcal{M} \rvert$.
          By (3) we have $\mathcal{M} \vDash \forall v (x \ne E[x/v] \vee \neg P[x/v])$, and thus $\mathcal{M}[v \mapsto \mathsf{y}] \vDash x \ne E[x/v] \vee \neg P[x/v]$.

          Let $\mathsf{z} = \llbracket E[x/v] \rrbracket_{\mathcal{M}[v \mapsto \mathsf{y}]}$. Then we have $\mathcal{M}[v \mapsto \mathsf{y}, x \mapsto \mathsf{z}] \nvDash x \ne E[x/v]$.

          Thus we have $\mathcal{M}[v \mapsto \mathsf{y}, x \mapsto \mathsf{z}] \vDash \neg P[x/v]$. Since $x \notin \fv(P[x/v])$, we have $\mathcal{M}[v \mapsto \mathsf{y}] \vDash \neg P[x/v]$. Thus we can conclude that $\mathcal{M} \vDash \forall v \neg P[x/v]$.
          
        \item[(5)] $\top \equiv \neg P[x/v] \equiv \neg P$ by logic using (4) and since $v \notin \fv(P)$
        \item[(6)] $\bot \equiv P$ by logic using (5)
      \end{description}

    \item[$C_1; C_2$]\leavevmode
      \begin{description}
        \item[(1)] $\bot \equiv \sp(P, C_1; C_2)$ by assumption
        \item[(2)] $\sp(P, C_1; C_2) \equiv \sp(\sp(P, C_1), C_2)$ by definition
        \item[(3)] $\bot \equiv \sp(\sp(P, C_1), C_2)$ by logic using (1) and (2)
        \item[(4)] $\bot \equiv \sp(P, C_1)$ by induction using (3)
        \item[(5)] $\bot \equiv P$ by induction using (4)
      \end{description}

    [$\ifstmt{E}{C_1}{C_2}$]\leavevmode
      \begin{description}
        \item[(1)] $\bot \equiv \sp(P, \ifstmt{E}{C_1}{C_2})$ by assumption
        \item[(2)] $\sp(P, \ifstmt{E}{C_1}{C_2}) \equiv \sp(P \wedge E, C_1) \vee \sp(P \wedge \neg E, C_2)$ by definition
        \item[(3)] $\bot \equiv \sp(P \wedge E, C_1) \vee \sp(P \wedge \neg E, C_2)$ by (1) and (2)
        \item[(4)] $\bot \equiv \sp(P \wedge E, C_1)$ by logic using (3)
        \item[(5)] $\bot \equiv \sp(P \wedge \neg E, C_2)$ by logic using (3)
        \item[(6)] $\bot \equiv P \wedge E$ by induction using (4)
        \item[(7)] $\bot \equiv P \wedge \neg E$ by induction using (5)
        \item[(8)] $\bot \equiv (P \wedge E) \vee (P \wedge \neg E) \equiv P \wedge (E \vee \neg E) \equiv P$ by logic using (6) and (7) \qedhere
      \end{description}
  \end{description}
\end{proof}

\begin{lemma}\label{lem:sp-sat}
  If $P \in \SatFormula$ then $\sp(P, C) \in \SatFormula$.
\end{lemma}
\begin{proof}
  We prove the contrapositive: assume $\sp(P, C) \notin \\ \SatFormula$, then $\sp(P, C) \equiv \bot$, and thus by \cref{lem:sp-is-false} $P \equiv \bot$. Therefore $P \notin \SatFormula$.
\end{proof}

\begin{lemma}\label{lem:sp-exists}
  If $y \notin \fv(C)$ then
  $\sp(\exists y\, P, C) \equiv \exists y \sp(P, C)$.
\end{lemma}
\begin{proof}
  By induction on $C$:
  \begin{description}
    \item[$\kskip$:] $\sp(\exists y\, P, \kskip) \equiv \exists y\, P \equiv \exists y\, \sp(P, \kskip)$
    \item[$x \coloneq E$:] Assuming $y \not\equiv v$,
      \begin{align*}
        &\sp(\exists y\, P, x \coloneq E) \\
          &\quad \equiv \exists v (x = E[x/v] \wedge (\exists y\, P)[x/v]) &\text{defn $\sp$}\\
          &\quad \equiv \exists v (x = E[x/v] \wedge (\exists y\, P[x/v])) &y \not\equiv x \in \fv(x \coloneq E) \\
          &\quad \equiv \exists y \exists v (x = E[x/v] \wedge P[x/v]) &y \notin \fv(E[x/v])\\
          &\quad \equiv \exists y \, \sp(P, x \coloneq E) &\text{defn $\sp$}
      \end{align*}
    \item[$C_1; C_2$:]
      \begin{align*}
        \sp(\exists y\, P, C_1; C_2) &\equiv \sp(\sp(\exists y\, P, C_1), C_2) &\text{defn $\sp$} \\
          &\equiv \sp(\exists y\, \sp(P, C_1), C_2) &\text{induction} \\
          &\equiv \exists y\, \sp(\sp(P, C_1), C_2) &\text{induction} \\
          &\equiv \exists y\, \sp(P, C_1; C_2) &\text{defn $\sp$}
      \end{align*}
    \item[$\ifstmt{E}{C_1}{C_2}$:]
      \begin{align*}
        &\sp(\exists y\, P, \ifstmt{E}{C_1}{C_2}) \\
        &\quad\equiv \sp((\exists y\, P) \wedge E, C_1) \vee \sp((\exists y\, P) \wedge \neg E, C_2) &\text{defn $\sp$} \\
        &\quad\equiv \sp(\exists y (P \wedge E), C_1) \vee \sp(\exists y (P \wedge \neg E), C_2) &y \notin \fv(E) \\
        &\quad\equiv (\exists y\, \sp(P \wedge E, C_1)) \vee (\exists y\, \sp(P \wedge \neg E, C_2)) &\text{induction} \\
        &\quad\equiv \exists y (\sp(P \wedge E, C_1) \vee \sp(P \wedge \neg E, C_2)) &\text{logic} \\
        &\quad\equiv \exists y \, \sp(P, \ifstmt{E}{C_1}{C_2}) &\text{defn $\sp$} &\qedhere
      \end{align*}
  \end{description}
\end{proof}

\begin{lemma}\label{lem:sp-disj}
  $\sp(P_1 \vee P_2, C) \equiv \sp(P_1, C) \vee \sp(P_2, C)$
\end{lemma}
\begin{proof}
  By induction on $C$:
  \begin{description}
    \item[$\kskip$:]
      $\sp(P_1 \vee P_2, \kskip) \equiv P_1 \vee P_2 \equiv \\ \sp(P_1, \kskip) \vee \sp(P_2, \kskip)$.
    \item[$x \coloneq E$:]
      \begin{align*}
        &\sp(P_1 \vee P_2, x \coloneq E) \\
          &\quad \equiv \exists v (x = E[x/v] \wedge (P_1 \vee P_2)[x/v]) &\text{defn $\sp$}\\
          &\quad \equiv \exists v (x = E[x/v] \wedge (P_1[x/v] \vee P_2[x/v])) &\text{subst.}\\
          &\quad \equiv \exists v ((x = E[x/v] \wedge P_1[x/v]) ~\vee \\
            &\hspace{3.9em} (x = E[x/v] \wedge P_2[x/v])) &\text{logic} \\
          &\quad \equiv (\exists v (x = E[x/v] \wedge P_1[x/v])) ~\vee\\
            &\hspace{2.4em} (\exists v (x = E[x/v] \wedge P_2[x/v])) &\text{logic}\\
          &\quad \equiv \sp(P_1, x \coloneq E) \vee \sp(P_2, x \coloneq E) &\text{defn $\sp$}
      \end{align*}

    \item[$C_1; C_2$:]
      \begin{align*}
        &\sp(P_1 \vee P_2, C_1; C_2) \\
          &\quad \equiv \sp(\sp(P_1 \vee P_2, C_1), C_2) &\text{defn $\sp$}\\
          &\quad \equiv \sp(\sp(P_1, C_1) \vee \sp(P_2, C_1), C_2) &\text{induction}\\
          &\quad \equiv \sp(\sp(P_1, C_1), C_2) \vee \sp(\sp(P_2, C_1), C_2) &\text{induction}\\
          &\quad \equiv \sp(P_1, C_1; C_2) \vee \sp(P_2, C_1; C_2) &\text{defn $\sp$}
      \end{align*}

    \item[$\ifstmt{E}{C_1}{C_2}$:]
      \begin{equation*}
        \begin{array}{p{1em}rlp{1em}r}
          \multicolumn{5}{l}{\sp(P_1 \vee P_2, \ifstmt{E}{C_1}{C_2})} \\
          &\equiv &\sp(E \wedge (P_1 \vee P_2), C_1) ~\vee \\
            &&\sp(\neg E \wedge (P_1 \vee P_2), C_2) &&\text{defn $\sp$} \\
          &\equiv &\sp((E \wedge P_1) \vee (E \wedge P_2), C_1) ~\vee \\
            &&\sp((\neg E \wedge P_1) \vee (\neg E \wedge P_2), C_2) &&\text{logic}\\
          &\equiv &\sp(E \wedge P_1, C_1) \vee \sp(E \wedge P_2, C_1) ~\vee \\
            &&\sp(\neg E \wedge P_1, C_2) \vee \sp(\neg E \wedge P_2, C_2) &&\text{\Cref{lem:sp-disj}} \\
          &\equiv &(\sp(E \wedge P_1, C_1) \vee \sp(\neg E \wedge P_1, C_2)) ~\vee \\
            &&(\sp(E \wedge P_2, C_1) \vee \sp(\neg E \wedge P_2, C_2)) &&\text{logic} \\
          &\equiv &\sp(P_1, \ifstmt{E}{C_1}{C_2}) ~\vee \\
            &&\sp(P_2, \ifstmt{E}{C_1}{C_2}) &&\text{defn $\sp$}
        \end{array}\qedhere
      \end{equation*}
  \end{description}

\end{proof}

\begin{lemma}[Frame rule]\label{lem:sp-frame}
  If $\mod(C) \cap \fv(P) = \emptyset$ then \\ $\sp(P \wedge R, C) \equiv P \wedge \sp(R, C)$.
\end{lemma}
\begin{proof}
  By induction on $C$:

  \begin{description}
    \item[$\kskip$:] $\sp(P \wedge R, \kskip) \equiv P \wedge R \equiv P \wedge \sp(R, \kskip)$
    
    \item[$x \coloneq E$:]
      Let $v \notin \fv(P)$. Note that $x \in \fv(x \coloneq E)$, thus from our assumptions $x \notin \fv(P)$ Then,
      \begin{align*}
        &\sp(P \wedge R, x \coloneq E) \\
          &\quad \equiv \exists v (x = E[x/v] \wedge P[x/v] \wedge R[x/v]) &\text{defn $\sp$}\\
          &\quad \equiv \exists v (x = E[x/v] \wedge P \wedge R[x/v]) &\text{$x \notin \fv(P)$} \\
          &\quad \equiv P \wedge \exists v(x = E[x/v] \wedge R[x/v]) &\text{$v \notin \fv(P)$} \\
          &\quad \equiv P \wedge \sp(R, x \coloneq E) &\text{defn $\sp$}
      \end{align*}
    
    \item[$C_1; C_2$:]
      Note that $x \notin \mod(C_1; C_2)$ thus $x \notin \mod(C_1)$ and $x \ClassNoteNoLine{class name}{note text} \mod(C_2)$. Then,
      \begin{align*}
        \sp(P \wedge R, C_1; C_2) &\equiv \sp(\sp(P \wedge R, C_1), C_2) &\text{defn $\sp$} \\
          &\equiv \sp(P \wedge \sp(R, C_1), C_2) &\text{induction} \\
          &\equiv P \wedge \sp(\sp(R, C_1), C_2) &\text{induction} \\
          &\equiv P \wedge \sp(R, C_1; C_2) &\text{defn $\sp$}
      \end{align*}

    \item[$\ifstmt{E}{C_1}{C_2}$:]
      Note that by assumption \\ $x \notin \mod(\ifstmt{E}{C_1}{C_2})$ thus $x \notin \mod(C_1)$ and $x \notin \mod(C_2)$. Then,
      \begin{align*}
        &\sp(P \wedge R, \ifstmt{E}{C_1}{C_2}) \\
          &\quad\equiv \sp(P \wedge R \wedge E, C_1) \vee \sp(P \wedge R \wedge \neg E, C_2) &\text{defn $\sp$} \\
          &\quad\equiv (P \wedge \sp(R \wedge E, C_1)) \vee (P \wedge \sp(R \wedge \neg E, C_2)) &\text{ind.} \\
          &\quad\equiv P \wedge (\sp(R \wedge E, C_1) \vee \sp(R \wedge \neg E, C_2)) &\text{logic} \\
          &\quad\equiv P \wedge \sp(R, \ifstmt{E}{C_1}{C_2}) &\text{defn $\sp$} &\qedhere
      \end{align*}
  \end{description}
\end{proof}

\begin{lemma}\label{lem:sp-if}
  \begin{align*}
    \sp(P \wedge E, \ifstmt{E}{C_1}{C_2}) &\equiv \sp(P \wedge E, C_1) \\
    \sp(P \wedge \neg E, \ifstmt{E}{C_1}{C_2}) &\equiv \sp(P \wedge \neg E, C_2)
  \end{align*}
\end{lemma}
\begin{proof}
  \begin{align*}
    &\sp(P \wedge E, \ifstmt{E}{C_1}{C_2}) \\
      &\quad\equiv \sp(P \wedge E \wedge E, C_1) \vee \sp(P \wedge E \wedge \neg E, C_2) &\text{defn $\sp$} \\
      &\quad\equiv \sp(P \wedge E, C_1) \vee \sp(\bot, C_2) &\text{logic} \\
      &\quad\equiv \sp(P \wedge E, C_1) \vee \bot &\text{\Cref{lem:sp-false}} \\
      &\quad\equiv \sp(P \wedge E, C_1) &\text{logic}
  \end{align*}
  \begin{align*}
    &\sp(P \wedge \neg E, \ifstmt{E}{C_1}{C_2}) \\
      &\quad\equiv \sp(P \wedge E \wedge \neg E, C_1) \vee \sp(P \wedge \neg E \wedge \neg E, C_2) &\text{defn sp} \\
      &\quad\equiv \sp(\bot, C_1) \vee \sp(P \wedge \neg E, C_2) &\text{logic} \\
      &\quad\equiv \bot \vee \sp(P \wedge \neg E, C_2) &\text{\Cref{lem:sp-false}} \\
      &\quad\equiv \sp(P \wedge \neg E, C_2) &\text{logic \qedhere}
  \end{align*}
\end{proof}

\subsection{Fixpoint}\label{app:fixpoint}

We can define a function $Q \mapsto \sp(\wp(C, Q), C)$. We demonstrate that this function reaches a fixpoint.

\begin{lemma}\label{lem:sp-wp-equiv}
  $\sp(\wp(C, Q) \wedge P, C) \equiv Q \wedge \sp(P, C)$
\end{lemma}
\begin{proof}
  By induction on $C$:
  \begin{description}
    \item[$\kskip$:]
      $\sp(\wp(\kskip, Q) \wedge P, \kskip) \equiv Q \wedge P \equiv \\ Q \wedge \sp(P, \kskip)$

    \item[$x \coloneq E$:]
      Let $v \notin \fv(Q)$. Then,
      \begin{align*}
        &\sp(\wp(x \coloneq E, Q) \wedge P, x \coloneq E) \\
          &\quad \equiv \exists v (x = E[x/v] ~\wedge \\
          &\hspace{3.7em} (\wp(x \coloneq E, Q) \wedge P)[x/v]) &\text{defn $\sp$}\\
          &\quad \equiv \exists v (x = E[x/v] \wedge (Q[x/E] \wedge P)[x/v]) &\text{defn $\wp$} \\
          &\quad \equiv \exists v (x = E[x/v] \wedge Q[x/E[x/v]] \wedge P[x/v]) &\text{subst} \\
          &\quad \equiv \exists v (x = E[x/v] \wedge Q[x/x] \wedge P[x/v]) &\text{subst =} \\
          &\quad \equiv \exists v (x = E[x/v] \wedge Q \wedge P[x/v]) &\text{redundant} \\
          &\quad \equiv Q \wedge (\exists v\, x = E[x/v] \wedge P[x/v]) &v \notin \fv(Q) \\
          &\quad \equiv Q \wedge \sp(P, x \coloneq E) &\text{defn $\sp$}
      \end{align*}

    \item[$\ifstmt{E}{C_1}{C_2}$:]\leavevmode
      \begin{description}
        \item[(1)] $Q \wedge \sp(P \wedge E, C_1) \equiv \sp(\wp(C_1, Q) \wedge P \wedge E, C_1)$ by induction using $C_1$
        \item[(2)] $Q \wedge \sp(P \wedge \neg E, C_2) \equiv \sp(\wp(C_2, Q) \wedge P \wedge \neg E, C_2)$ by induction using $C_2$
        \item[(3)]
          \begin{align*}
            &Q \wedge \sp(P, \ifstmt{E}{C_1}{C_2}) \\
            &\quad \equiv Q \wedge (\sp(P \wedge E, C_1) \vee \sp(P \wedge \neg E, C_2)) &\text{defn $\sp$} \\
            &\quad \equiv (Q \wedge \sp(P \wedge E, C_1)) ~\vee \\
              &\hspace{2.35em} (Q \wedge \sp(P \wedge \neg E, C_2)) &\text{logic}\\
            &\quad \equiv \sp(\wp(C_1, Q) \wedge P \wedge E, C_1) ~\vee \\
            &\hspace{2.1em} \sp(\wp(C_2, Q) \wedge P \wedge \neg E, C_2) &\text{(1), (2)}\\
            &\quad \equiv \sp(\wp(C_1, Q) \wedge P \wedge E, \\
              &\hspace{3.6em} \ifstmt{E}{C_1}{C_2}) ~\vee \\
            &\hspace{2.25em} \sp(\wp(C_2, Q) \wedge P \wedge \neg E, \\
              &\hspace{3.6em} \ifstmt{E}{C_1}{C_2}) &\text{\Cref{lem:sp-if}}\\
            &\quad\equiv \sp(((\wp(C_1, Q) \wedge E) ~\vee \\
              &\hspace{4em} (\wp(C_2, Q) \wedge \neg E)) \wedge P, \\
            &\hspace{3.8em} \ifstmt{E}{C_1}{C_2}) &\text{\Cref{lem:sp-disj}}\\
            &\quad\equiv \sp(\wp(\ifstmt{E}{C_1}{C_2}, Q) \wedge P, \\
            &\hspace{3.5em} \ifstmt{E}{C_1}{C_2}) &\text{defn $\wp$}
          \end{align*}
      \end{description}

    \item[$C_1; C_2$:]\leavevmode
      \begin{description}
        \item[(1)] $\wp(C_2, Q) \wedge \sp(P, C_1) \equiv \sp(\wp(C_1, \\ \wp(C_2, Q)) \wedge P, C_1)$ by induction using $C_1$
        \item[(2)] $Q \wedge \sp(\sp(P, C_1), C_2) \equiv \sp(\wp(C_2, Q) \wedge  \\ \sp(P, C_1), C_2)$ by induction using $C_2$. Then,
      \end{description}
      \begin{align*}
        &Q \wedge \sp(P, C_1; C_2) \\
          &\quad\equiv Q \wedge \sp(\sp(P, C_1), C_2) &\text{defn $\sp$}\\
          &\quad\equiv \sp(\wp(C_2, Q) \wedge \sp(P, C_1), C_2) &\text{(2)} \\
          &\quad\equiv \sp(\sp(\wp(C_1, \wp(C_2, Q)) \wedge P, C_1), C_2) &\text{(1)} \\
          &\quad\equiv \sp(\sp(\wp(C_1; C_2, Q) \wedge P, C_1), C_2) &\text{defn $\wp$}\\
          &\quad\equiv \sp(\wp(C_1; C_2, Q) \wedge P, C_1; C_2) &\text{defn $\sp$} &\qedhere
      \end{align*}
  \end{description}
\end{proof}

\begin{lemma}[Fixpoint property]\label{lem:sp-wp-fixpoint}
  If $Q \Rightarrow \sp(\top, C)$ then \\ $Q \equiv \sp(\wp(C, Q), C)$.
\end{lemma}
\begin{proof}
  Immediate from \cref{lem:sp-wp-equiv}, noting that $Q \equiv Q \wedge \sp(\top, C)$ in this case.
\end{proof}

\begin{lemma}\label{lem:sp-wp-fixpoint-p}
  If $Q \Rightarrow \sp(P, C)$ then $Q \equiv \sp(P \wedge \wp(C, Q), C)$.
\end{lemma}
\begin{proof}
  \begin{align*}
    Q &\equiv Q \wedge \sp(P, C) &Q \Rightarrow \sp(P, C) \\
      &\equiv \sp(P \wedge \wp(C, Q), C) &\text{\Cref{lem:sp-wp-equiv}} &\qedhere
  \end{align*}
\end{proof}

\section{Hoare logic}\label{app:hoare}

Hoare triples are denoted by $\htrip{P}{C}{Q}$. Deductions in Hoare logic are characterized by the following rules:

\begin{mathpar}
  \inferrule[OX-Skip]
    { }
    { \htrip{P}{\ttt{skip}}{P} } \and
  \inferrule[OX-Assign]
    { }
    { \htrip{P[x/E]}{x \coloneq E}{P} } \and
  \inferrule[OX-Seq]
    { \htrip{P}{C_1}{R} \\\\ \htrip{R}{C_2}{Q} }
    { \htrip{P}{C_1; C_2}{Q} } \and
  \inferrule[OX-If]
    { \htrip{E \wedge P}{C_1}{Q} \\\\ \htrip{\neg E \wedge P}{C_2}{Q} }
    { \htrip{P}{\ttt{if}~ E ~\ttt{then}~ C_1 ~\ttt{else}~ C_2}{Q} } \and
  \inferrule[OX-Cons]
    { P \Rightarrow P' \\ \htrip{P'}{C}{Q'} \\ Q' \Rightarrow Q }
    { \htrip{P}{C}{Q} }
\end{mathpar}

\subsection{Weakest preconditions}\label{app:hoare-wp}

\begin{lemma}\label{lem:hoare-imply-wp}
  If $\htrip{P}{C}{Q}$ then $P \Rightarrow \wp(C, Q)$.
\end{lemma}
\begin{proof}
  By induction on the derivation of $\htrip{P}{C}{Q}$:

  \begin{description}
    \item[OX-Skip:]\leavevmode
      \begin{description}
        \item[(1)] $P \equiv Q$ by inversion
        \item[(2)] $P \equiv Q \equiv \wp(\kskip, Q)$ by definition
      \end{description}
    \item[OX-Assign:]\leavevmode
      \begin{description}
        \item[(1)] $P \equiv Q[E/x]$ by inversion
        \item[(2)] $P \equiv Q[E/x] \equiv \wp(x \coloneq E, Q)$ by definition
      \end{description}
    \item[OX-Seq:]\leavevmode
      \begin{description}
        \item[(1)] $\htrip{P}{C_1}{R}$ for some $R$ by inversion
        \item[(2)] $P \Rightarrow \sp(C_1, R)$ by induction using (1)
        \item[(3)] $\htrip{R}{C_2}{Q}$ by inversion
        \item[(4)] $R \Rightarrow \wp(C_2, Q)$ by induction using (3)
        \item[(5)] $\wp(C_1, R) \Rightarrow \wp(C_1, \wp(C_2, Q))$ by \cref{lem:wp-strengthen-post} using (4)
        \item[(6)] By (4), (5), and definition of $\wp$,
          \begin{align*}
            P &\Rightarrow \wp(C_1, R) \\
              &\Rightarrow \wp(C_1, \wp(C_2, Q)) \\
              &\equiv \wp(C_1; C_2, Q)
          \end{align*}
      \end{description}
    \item[OX-If:]\leavevmode
      \begin{description}
        \item[(1)] $\htrip{E \wedge P}{C_1}{Q}$ by inversion
        \item[(2)] $E \wedge P \Rightarrow \wp(C_1, Q)$ by induction using (1)
        \item[(3)] $\htrip{\neg E \wedge P}{C_2}{Q}$ by inversion
        \item[(4)] $\neg E \wedge P \Rightarrow \wp(C_2, Q)$ by induction using (3)
        \item[(5)]
          \begin{align*}
            P &\equiv (P \wedge E) \vee (P \wedge \neg E) &\text{logic} \\
              &\Rightarrow (\wp(C_1, Q) \wedge E) \vee (\wp(C_2, Q) \wedge \neg E) &\text{(2) and (4)} \\
              &\Rightarrow \wp(\ifstmt{E}{C_1}{C_2}, Q) &\text{defn $\wp$}
          \end{align*}
      \end{description}

    \item[OX-Cons:]\leavevmode
      \begin{description}
        \item[(1)] $P \Rightarrow P'$ for some $P'$ by inversion
        \item[(2)] $Q' \Rightarrow Q$ for some $Q'$ by inversion
        \item[(3)] $\htrip{P'}{C}{Q'}$ by inversion
        \item[(4)] $P' \Rightarrow \wp(C, Q')$ by induction using (3)
        \item[(5)] $\wp(C, Q') \Rightarrow \wp(C, Q)$ by \cref{lem:wp-strengthen-post} using (4)
        \item[(6)] $P \Rightarrow P' \Rightarrow \wp(C, Q') \Rightarrow \wp(C, Q)$ by (1), (4), and (5) \qedhere
      \end{description}
  \end{description}
\end{proof}

\begin{lemma}\label{lem:wp-provable}
  $\htrip{\wp(C, Q)}{C}{Q}$
\end{lemma}
\begin{proof}
  By induction on $C$:

  \begin{description}
    \item[$\kskip$:] $\wp(\kskip, Q) \equiv Q$ by definition; $\etrip{Q}{\kskip}{Q}$ by \textsc{OX-Skip}.
    \item[$x \coloneq E$:]\leavevmode
      \begin{description}
        \item[(1)] $\wp(x \coloneq E, Q) \equiv Q[x/E]$ by definition
        \item[(2)] $\htrip{Q[x/E]}{x \coloneq E}{Q}$ by \textsc{OX-Assign} using (1)
      \end{description}
    \item[$C_1; C_2$:]\leavevmode
      \begin{description}
        \item[(1)] $\htrip{\wp(C_1, \wp(C_2, Q))}{C_1}{\wp(C_2, Q)}$ by induction
        \item[(2)] $\htrip{\wp(C_2, Q)}{C_2}{Q}$ by induction
        \item[(3)] $\htrip{\wp(C_1, \wp(C_2, Q))}{C_1; C_2}{Q}$ by \textsc{OX-Seq} using (1) and (2)
        \item[(4)] $\wp(C_1; C_2, Q) \equiv \wp(C_1, \wp(C_2, Q))$ by definition
        \item[(5)] $\htrip{\wp(C_1; C_2, Q)}{C_1; C_2}{Q}$ by (3) and (4)
      \end{description}
    \item[$\ifstmt{E}{C_1}{C_2}$:]\leavevmode
      \begin{description}
        \item[(1)] Let $P \equiv (\wp(C_1, Q) \wedge E) \vee (\wp(C_2, Q) \wedge \neg E)$
        \item[(2)] $\htrip{\wp(C_1, Q)}{C_1}{Q}$ by induction
        \item[(3)] $P \wedge E \Rightarrow \wp(C_1, Q)$ by logic
        \item[(4)] $\htrip{P \wedge E}{C_1}{Q}$ by \textsc{OX-Cons} using (2) and (3)
        \item[(5)] $\htrip{\wp(C_2, Q)}{C_2}{Q}$ by induction
        \item[(6)] $P \wedge \neg E \Rightarrow \wp(C_2, Q)$ by logic
        \item[(7)] $\htrip{P \wedge \neg E}{C_2}{Q}$ by \textsc{OX-Cons} using (5) and (6)
        \item[(8)] $\htrip{P}{\ifstmt{E}{C_1}{C_2}}{Q}$ by \textsc{OX-If} using (4) and (7)
        \item[(9)] $P \equiv \wp(\ifstmt{E}{C_1}{C_2}, Q)$ by (1) and definition of $\wp$
        \item[(10)] $\vdash \{ \wp(\ifstmt{E}{C_1}{C_2}) \} \\ \ifstmt{E}{C_1}{C_2}\ \{Q\}$ by (8) and (9) \qedhere
      \end{description}
  \end{description}
\end{proof}

\begin{lemma}\label{lem:wp-imply-hoare}
  If $P \Rightarrow \wp(C, Q)$ then $\htrip{P}{C}{Q}$.
\end{lemma}
\begin{proof}
  By \cref{lem:wp-provable}, $\htrip{\wp(C, Q)}{C}{Q}$, thus by \textsc{OX-Cons} $\htrip{P}{C}{Q}$.
\end{proof}

\begin{theorem}\label{thm:wp-iff-hoare}
  $P \Rightarrow \wp(C, Q) \iff \htrip{P}{C}{Q}$.
\end{theorem}
\begin{proof}
  $\Longrightarrow$: \Cref{lem:wp-imply-hoare}; $\Longleftarrow$: \Cref{lem:hoare-imply-wp}.
\end{proof}

\subsection{Strongest postconditions}\label{app:hoare-sp}

\begin{lemma}\label{lem:hoare-imply-sp}
  If $\htrip{P}{C}{Q}$ then $\sp(P, C) \Rightarrow Q$.
\end{lemma}
\begin{proof}
  By induction on $\htrip{P}{C}{Q}$:

  \begin{description}
    \item[OX-Skip:] Then $P \equiv Q$ and $C \equiv \kskip$, thus \\ $\sp(P, \kskip) \equiv P \equiv Q$.
    
    \item[OX-Assign:] Then $P \equiv Q[x/E]$ and $C \equiv x \coloneq E$. Assuming $v \notin \fv(Q)$:
    \begin{align*}
      &\sp(Q[x/E], x \coloneq E) \\
        &\quad \equiv \exists v (x = E[x/v] \wedge Q[x/E][x/v]) &\text{defn} \\
        &\quad \equiv \exists v (x = E[x/v] \wedge Q[x/E[x/v])) &\text{substitution}\\
        &\quad \equiv \exists v (x = E[x/v] \wedge Q[x/x]) &\text{subst.\ equality} \\
        &\quad \equiv \exists v (x = E[x/v] \wedge Q) &\text{redundant}\\
        &\quad \Rightarrow Q &v \notin \fv(Q)
    \end{align*}

    \item[OX-Seq:]\leavevmode
      \begin{description}
        \item[(1)] $C \equiv C_1; C_2$ for some $C_1, C_2$ by inversion
        \item[(2)] $\htrip{P}{C_1}{R}$ for some $R$ by inversion
        \item[(3)] $\htrip{R}{C_2}{Q}$ by inversion
        \item[(4)] $\sp(P, C_1) \Rightarrow R$ by induction using (2)
        \item[(5)] $\sp(\sp(P, C_1), C_2) \Rightarrow \sp(R, C_2)$ by \cref{lem:sp-strengthen} using (4)
        \item[(6)] $\sp(R, C_2) \Rightarrow Q$ by induction using (3)
        \item[(7)] $\sp(\sp(P, C_1), C_2) \Rightarrow Q$ by (5) and (6)
        \item[(8)] $\sp(P, C) \equiv \sp(P, C_1; C_2) \Rightarrow Q$ by (1), (7), and definition of $\sp$
      \end{description}

    \item[OX-If:]\leavevmode
    \begin{description}
      \item[(1)] $C \equiv \ifstmt{E}{C_1}{C_2}$ for some $E, C_1, C_2$ by inversion
      \item[(2)] $\htrip{E \wedge P}{C_1}{Q}$ by inversion
      \item[(3)] $\htrip{\neg E \wedge P}{C_2}{Q}$ by inversion
      \item[(4)] $\sp(E \wedge P, C_1) \Rightarrow Q$ by induction using (2)
      \item[(5)] $\sp(\neg E \wedge P, C_2) \Rightarrow Q$ by induction using (3)
      \item[(6)] $\sp(E \wedge P, C_1) \vee \sp(\neg E \wedge P, C_2) \Rightarrow Q$ by logic using (4) and (5)
      \item[(7)] $\sp(P, \ifstmt{E}{C_1}{C_2}) \Rightarrow Q$ by (6) and definition of $\sp$
    \end{description}

    \item[OX-Cons:]\leavevmode
      \begin{description}
        \item[(1)] $P \Rightarrow P'$ for some $P'$ by inversion
        \item[(2)] $Q' \Rightarrow Q$ for some $Q'$ by inversion
        \item[(3)] $\htrip{P'}{C}{Q'}$ by inversion
        \item[(4)] $\sp(P, C) \Rightarrow \sp(P', C)$ by \cref{lem:sp-strengthen} and (1)
        \item[(5)] $\sp(P', C) \Rightarrow Q'$ by induction using (3)
        \item[(6)] $\sp(P, C) \Rightarrow Q$ by (4), (5), and (2) \qedhere
      \end{description}
  \end{description}
\end{proof}

\begin{lemma}\label{lem:sp-provable}
  $\htrip{P}{C}{\sp(P, C)}$
\end{lemma}
\begin{proof}
  By induction on $C$:

  \begin{description}
    \item[$\kskip$:]\leavevmode
      \begin{description}
        \item[(1)] $\htrip{P}{\kskip}{P}$ by \textsc{OX-Skip}
        \item[(2)] $P \equiv \sp(P, \kskip)$ by \textsc{OX-Skip}
        \item[(3)] $\htrip{P}{\kskip}{\sp(P, \kskip)}$ by \textsc{OX-Cons} using (1) and (2)
      \end{description}

    \item[$x \coloneq E$:]\leavevmode
      \begin{description}
        \item[(1)] Let $Q \equiv \exists v (x = E[x/v] \wedge P[x/v])$ where $v \notin \fv(P)$.
        \item[(2)] $\htrip{Q[E/x]}{x \coloneq E}{Q}$ by \textsc{OX-Assign}
        \item[(3)] $P \Rightarrow \exists v (E = E[x/v] \wedge P[x/v]) \equiv Q[x/E]$ by logic (witnessed by letting $v = x$).
        \item[(4)] $\htrip{P}{x \coloneq E}{Q}$ by \textsc{OX-Cons} using (2) and (3)
        \item[(5)] $\htrip{P}{x \coloneq E}{\sp(P, x \coloneq E)}$ by (1), (4), and definition of $\sp$
      \end{description}

    \item[$C_1; C_2$:]\leavevmode
    \begin{description}
      \item[(1)] $\htrip{P}{C_1}{\sp(P, C_1)}$ by induction
      \item[(2)] $\htrip{\sp(P, C_1)}{C_2}{\sp(\sp(P, C_1), C_2)}$ by induction
      \item[(3)] $\htrip{P}{C_1; C_2}{\sp(\sp(P, C_1), C_2)}$ by \textsc{OX-Seq} using (1) and (2)
      \item[(4)] $\sp(P, C_1; C_2) \equiv \sp(\sp(P, C_1), C_2)$ by definition
      \item[(5)] $\htrip{P}{C_1; C_2}{\sp(P, C_1; C_2)}$ by (3) and (4)
    \end{description}

    \item[$\ifstmt{E}{C_1}{C_2}$:]\leavevmode
    \begin{description}
      \item[(1)] $\htrip{P \wedge E}{C_1}{\sp(P \wedge E, C_1)}$ by induction
      \item[(2)] $\htrip{P \wedge \neg E}{C_2}{\sp(P \wedge \neg E, C_2)}$ by induction
      \item[(3)] $\htrip{P \wedge E}{C_1}{\sp(P \wedge E, C_1) \vee \sp(P \wedge \neg E, C_2)}$ by \textsc{OX-Cons} using (1)
      \item[(4)] $\htrip{P \wedge \neg E}{C_2}{\sp(P \wedge E, C_1) \vee \sp(P \wedge \neg E, C_2)}$ by \textsc{OX-Cons} using (2)
      \item[(5)] $\htrip{P}{\ifstmt{E}{C_1}{C_2}}{\sp(P \wedge E, C_1) \vee \sp(P \wedge \neg E, C_2)}$ by \textsc{OX-If} using (3) and (4)
      \item[(6)] $\vdash \{P\} ~ \ifstmt{E}{C_1}{C_2} \\ \{ \sp(P, \ifstmt{E}{C_1}{C_2}) \}$ by (5) and definition of $\sp$ \qedhere
    \end{description}
  \end{description}
\end{proof}

\begin{lemma}\label{lem:sp-imply-hoare}
  If $\sp(P, C) \Rightarrow Q$ then $\htrip{P}{C}{Q}$.
\end{lemma}
\begin{proof}
  Immediate from \cref{lem:sp-provable}, applying \textsc{OX-Cons}.
\end{proof}

\begin{theorem}\label{thm:hoare-iff-sp}
  $\sp(P, C) \Rightarrow Q \iff \htrip{P}{C}{Q}$
\end{theorem}
\begin{proof}
  $\Longrightarrow$: \Cref{lem:sp-imply-hoare}; $\Longleftarrow$: \Cref{lem:hoare-imply-sp}
\end{proof}

\section{Incorrectness logic}\label{app:il}

Incorrectness logic triples are denoted $\itrip{P}{C}{Q}$. Deductions in incorrectness logic are defined as follows:
\begin{mathpar}
  \inferrule[UX-Skip]
    { }
    { \itrip{P}{\ttt{skip}}{P} } \and
  \inferrule[UX-Seq]
    { \itrip{P}{C_1}{Q} \\ \itrip{Q}{C_2}{R} }
    { \itrip{P}{C_1; C_2}{R} } \and
  \inferrule[UX-Assign]
    { }
    { \itrip{P}{x \coloneq E}{\exists v (x = E[x/v] \wedge P[x/v])} } \and
  \inferrule[UX-IfThen]
    { \itrip{E \wedge P}{C_1}{Q} }
    { \itrip{E \wedge P}{\ifstmt{E}{C_1}{C_2}}{Q} } \and
  \inferrule[UX-IfElse]
    { \itrip{\neg E \wedge P}{C_2}{Q} }
    { \itrip{\neg E \wedge P}{\ifstmt{E}{C_1}{C_2}}{Q} } \\
  \inferrule[UX-Cons]
    { P' \Rightarrow P \\\\ \itrip{P'}{C}{Q'} \\\\ Q \Rightarrow Q' }
    { \itrip{P}{C}{Q} } \and
  \inferrule[UX-Disj]
    { \itrip{P_1}{C}{Q_1} \\\\
      \itrip{P_2}{C}{Q_2} }
    { \itrip{P_1 \vee P_2}{C}{Q_1 \vee Q_2} }
\end{mathpar}

\subsection{Strongest postconditions}

\begin{lemma}\label{lem:il-sp-provable}
  $\itrip{P}{C}{\sp(P, C)}$
\end{lemma}
\begin{proof}
  By induction on $C$:
  \begin{description}
    \item[$\kskip$:] By \textsc{UX-Skip} $\itrip{P}{\kskip}{P}$ and $P \equiv \sp(P, \kskip)$ by definition.
    
    \item[$x \coloneq E$:] \leavevmode
    \begin{description}
      \item[(1)] $\itrip{P}{x \coloneq E}{\exists v (x = E[x/v] \wedge P[x/v])}$ by \textsc{UX-Assign}
      \item[(2)] $\sp(P, x \coloneq E) \equiv \exists v (x = E[x/v] \wedge P[x/v])$ by definition of $\sp$
      \item[(3)] $\itrip{P}{x \coloneq E}{sp(P, x \coloneq E)}$ by (1) and (2)
    \end{description}
    
    \item[$C_1; C_2$:]\leavevmode
      \begin{description}
        \item[(1)] $\itrip{P}{C_1}{\sp(P, C_1)}$ by induction
        \item[(2)] $\itrip{\sp(P, C_1)}{C_2}{\sp(\sp(P, C_1), C_2)}$ by induction
        \item[(3)] $\itrip{P}{C_1; C_2}{\sp(\sp(P, C_1), C_2)}$ by \textsc{UX-Seq} using (1), (2)
        \item[(4)] $\sp(P, C_1; C_2) \equiv \sp(\sp(P, C_1), C_2)$ by definition
        \item[(5)] $\itrip{P}{C_1; C_2}{\sp(P, C_1; C_2)}$ by (3), (4)
      \end{description}

    \item[$\ifstmt{E}{C_1}{C_2}$:]\leavevmode
      \begin{description}
        \item[(1)] $\itrip{E \wedge P}{C_1}{\sp(E \wedge P, C_1)}$ by induction
        \item[(2)] $\itrip{\neg E \wedge P}{C_2}{\sp(\neg E \wedge P, C_2)}$ by induction
        \item[(3)] $\itrip{E \wedge P}{\ifstmt{E}{C_1}{C_2}}{\sp(E \wedge P, C_1)}$ by \textsc{UX-IfThen} using (1)
        \item[(4)] $\itrip{\neg E \wedge P}{\ifstmt{E}{C_1}{C_2}}{\sp(\neg E \wedge P, C_2)}$ by \textsc{UX-IfElse} using (2)
        \item[(5)] $\itrip{P}{\ifstmt{E}{C_1}{C_2}}{\sp(E \wedge P, C_1) \vee \sp(\neg E \wedge P, C_2)}$ by \textsc{UX-Disj} using (3) and (4)
        \item[(6)] $\sp(P, \ifstmt{E}{C_1}{C_2}) \equiv \sp(E \wedge P, C_1) \vee \sp(\neg E \wedge P, C_2)$ by definition
        \item[(7)] $\vdash [P] ~\ifstmt{E}{C_1}{C_2} \\{} [\sp(P, \ifstmt{E}{C_1}{C_2})]$ by (5) and (6) \qedhere
      \end{description}
  \end{description}
\end{proof}

\begin{lemma}\label{lem:sp-to-il}
  If $Q \Rightarrow \sp(P, C)$ then $\itrip{P}{C}{Q}$.
\end{lemma}
\begin{proof}
  Immediate from \cref{lem:il-sp-provable} and \textsc{UX-Cons}.
\end{proof}

\begin{lemma}\label{lem:il-to-sp}
  If $\itrip{P}{C}{Q}$ then $Q \Rightarrow \sp(P, C)$.
\end{lemma}
\begin{proof}
  By induction on the derivation $\itrip{P}{C}{Q}$:
  \begin{description}
    \item[UX-Skip:]\leavevmode
      \begin{description}
        \item[(1)] $Q \equiv P$ by inversion
        \item[(2)] $\sp(P, \kskip) \equiv P$ by definition
        \item[(3)] $Q \equiv \sp(P, \kskip)$ by (1) and (2)
      \end{description}

    \item[UX-Assign:]\leavevmode
      \begin{description}
        \item[(1)] $Q \equiv \exists v (x = E[x/v] \wedge P[x/v])$ by inversion
        \item[(2)] $\sp(P, x \coloneq E) \equiv \exists v (x = E[x/v] \wedge P[x/v])$ by definition
        \item[(3)] $Q \equiv \sp(P, x \coloneq E)$ by (1) and (2)
      \end{description}

    \item[UX-Seq:]\leavevmode
      \begin{description}
        \item[(1)] $\itrip{P}{C_1}{R}$ for some $R$ by inversion
        \item[(2)] $\itrip{R}{C_2}{Q}$ by inversion
        \item[(3)] $Q \Rightarrow \sp(R, C_2)$ by induction using (2)
        \item[(4)] $R \Rightarrow \sp(P, C_1)$ by induction using (1)
        \item[(5)] $\sp(R, C_2) \Rightarrow \sp(\sp(P, C_1), C_2)$ by \cref{lem:sp-strengthen} using (4)
        \item[(6)]
          \begin{align*}
            Q &\Rightarrow \sp(R, C_2) &\text{(3)} \\
              &\Rightarrow \sp(\sp(P, C_1), C_2) &\text{(5)} \\
              &\equiv \sp(P, C_1; C_2) &\text{defn $\sp$}
          \end{align*}
      \end{description}
      
    \item[UX-IfThen:]\leavevmode
      \begin{description}
        \item[(1)] $C \equiv \ifstmt{E}{C_1}{C_2}$ for some $E, C_1, C_2$ by inversion
        \item[(2)] $P \equiv E \wedge P'$ for some $P'$ by inversion
        \item[(3)] $\itrip{E \wedge P'}{C_1}{Q}$ by inversion
        \item[(4)] $Q \Rightarrow \sp(E \wedge P', C_1)$ by induction using (3)
        \item[(5)] $\sp(E \wedge P', C_1) \equiv \sp(E \wedge P', \ifstmt{E}{C_1}{C_2})$ by \cref{lem:sp-if}
        \item[(6)]
          \begin{align*}
            Q &\Rightarrow \sp(E \wedge P', C_1) &\text{(4)} \\
              &\equiv \sp(E \wedge P', \ifstmt{E}{C_1}{C_2}) &\text{(5)} \\
              &\equiv \sp(P, \ifstmt{E}{C_1}{C_2}) &\text{(2)}
          \end{align*}
      \end{description}
    \item[UX-IfElse:]\leavevmode
      \begin{description}
        \item[(1)] $C \equiv \ifstmt{E}{C_1}{C_2}$ for some $E, C_1, C_2$ by inversion
        \item[(2)] $P \equiv \neg E \wedge P'$ for some $P'$ by inversion
        \item[(3)] $\itrip{\neg E \wedge P'}{C_2}{Q}$ by inversion
        \item[(4)] $Q \Rightarrow \sp(\neg E \wedge P', C_2)$ by induction using (3)
        \item[(5)] $\sp(\neg E \wedge P', C_2) \equiv \sp(\neg E \wedge P', \\ \ifstmt{E}{C_1}{C_2})$ by \cref{lem:sp-if}
        \item[(6)]
          \begin{align*}
            Q &\Rightarrow \sp(\neg E \wedge P', C_2) &\text{(4)} \\
              &\equiv \sp(\neg E \wedge P', \ifstmt{E}{C_1}{C_2}) &\text{(5)} \\
              &\equiv \sp(P, \ifstmt{E}{C_1}{C_2}) &\text{(2)}
          \end{align*}
      \end{description}
    \item[UX-Cons:]\leavevmode
      \begin{description}
        \item[(1)] $\itrip{P'}{C}{Q'}$ for some $P', Q'$ by inversion
        \item[(2)] $P' \Rightarrow P$ by inversion
        \item[(3)] $Q \Rightarrow Q'$ by inversion
        \item[(4)] $Q' \Rightarrow \sp(P', C)$ by induction using (1)
        \item[(5)] $\sp(P', C) \Rightarrow \sp(P, C)$ by \cref{lem:sp-strengthen} using (2)
        \item[(6)] $Q \Rightarrow Q' \Rightarrow \sp(P', C) \Rightarrow \sp(P, C)$ by (3), (4), and (5)
      \end{description}
    \item[UX-Disj:]\leavevmode
      \begin{description}
        \item[(1)] $\itrip{P_1}{C}{Q_1}$ for some $P_1, Q_1$ by inversion
        \item[(2)] $\itrip{P_2}{C}{Q_2}$ for some $P_2, Q_2$ by inversion
        \item[(3)] $P \equiv P_1 \vee P_2$ by inversion
        \item[(4)] $Q \equiv Q_1 \vee Q_2$ by inversion
        \item[(5)] $Q_1 \Rightarrow \sp(P_1, C)$ by induction using (1)
        \item[(6)] $Q_2 \Rightarrow \sp(P_2, C)$ by induction using (2)
        \item[(7)] $Q_1 \vee Q_2 \Rightarrow \sp(P_1, C) \vee \sp(P_2, C)$ by logic using (5) and (6)
        \item[(8)]
          \begin{align*}
            Q &\equiv Q_1 \vee Q_2 &\text{(4)} \\
              &\Rightarrow \sp(P_1, C) \vee \sp(P_2, C) &\text{(7)} \\
              &\equiv \sp(P_1 \vee P_2, C) &\text{\Cref{lem:sp-disj}} \\
              &\equiv \sp(P, C) &\text{(3)} &\qedhere
          \end{align*}
      \end{description}
  \end{description}
\end{proof}

\begin{theorem}\label{thm:il-iff-sp}
  $\itrip{P}{C}{Q} \iff Q \Rightarrow \sp(P, C)$
\end{theorem}
\begin{proof}
  $\Longrightarrow$: \Cref{lem:il-to-sp}; $\Longleftarrow$: \Cref{lem:sp-to-il}.
\end{proof}

\subsection{Satisfiability}

\begin{lemma}\label{lem:il-pre-false}
  If $\itrip{\bot}{C}{Q}$ then $Q \equiv \bot$.
\end{lemma}
\begin{proof}\leavevmode
  \begin{description}
    \item[(1)] $Q \Rightarrow \sp(\bot, C)$ by \cref{thm:hoare-iff-sp}
    \item[(2)] $\sp(\bot, C) \equiv \bot$ by \cref{lem:sp-false}
    \item[(3)] $Q \Rightarrow \bot$ logic using (1) and (2)
    \item[(4)] $Q \equiv \bot$ by logic using (3) \qedhere
  \end{description}
\end{proof}

\begin{lemma}\label{lem:il-sat}
  If $\itrip{P}{C}{Q}$ and $Q \in \SatFormula$ then $P \in \SatFormula$.
\end{lemma}
\begin{proof}
  Assume $\itrip{P}{C}{Q}$, then we prove the contrapositive; that is, $P \notin \SatFormula \implies Q \notin \SatFormula$.

  Assuming $P \notin \SatFormula$, since $P \in \Formula$ we get $P \equiv \bot$, and thus $\itrip{\bot}{C}{Q}$ by assumption. Then by \cref{lem:il-pre-false}, $Q \equiv \bot$ thus $Q \notin \SatFormula$.
\end{proof}

\begin{lemma}\label{lem:il-wp-sat}
  If $\itrip{P}{C}{Q}$ and $Q \in \SatFormula$ then $P \wedge \wp(P, C) \in \SatFormula$.
\end{lemma}
\begin{proof}\leavevmode
  \begin{description}
    \item[(1)] $\itrip{P}{C}{Q}$ by assumption
    \item[(2)] $Q \in \SatFormula$ by assumption
    \item[(3)] $Q \Rightarrow \sp(P, C)$ by \cref{lem:il-to-sp} using (1)
    \item[(4)] $\sp(P \wedge \wp(P, C), C) \equiv Q$ by \cref{lem:sp-wp-fixpoint-p} using (3)
    \item[(5)] $\itrip{P \wedge \wp(P, C)}{C}{Q}$ by \cref{lem:sp-to-il} using (4)
    \item[(6)] $P \wedge \wp(C, Q) \in \SatFormula$ by \cref{lem:il-sat} using (2) and (5) \qedhere
  \end{description}
\end{proof}

\section{Exact logic}\label{app:el}

Exact logic triples are denoted by $\etrip{P}{C}{Q}$. Deductions in exact logic are characterized by the following rules:
\begin{mathpar}
  \inferrule[EX-Skip]
    { }
    { \etrip{\top}{\ttt{skip}}{\top} } \hfill
  \inferrule[EX-Assign]
    { x \notin \fv(E') }
    { \etrip{x = E'}{x \coloneq E}{x = E[x/E']} } \\
  \inferrule[EX-IfThen]
    { \etrip{P \wedge E}{C_1}{Q} }
    { \etrip{P \wedge E}{\ifstmt{E}{C_1}{C_2}}{Q} } \hfill
    \inferrule[EX-Seq]
    { \etrip{P}{C_1}{R} \\\\ \etrip{R}{C_2}{Q} }
    { \etrip{P}{C_1; C_2}{Q} } \\
    \inferrule[EX-IfElse]
    { \etrip{P \wedge \neg E}{C_2}{Q} }
    { \etrip{P \wedge \neg E}{\ifstmt{E}{C_1}{C_2}}{Q} } \hfill
  \inferrule[EX-Frame]
    { \operatorname{mod}(C) \cap \fv(R) = \emptyset \\\\
      \etrip{P}{C}{Q} }
    { \etrip{P \wedge R}{C}{Q \wedge R} } \\
  \inferrule[EX-Exists]
    { \etrip{P}{C}{Q} \\\\ x \notin \fv(C)}
    { \etrip{\exists x\, P}{C}{\exists x\, Q} } \hfill
  \inferrule[EX-Disj]
    { \etrip{P_1}{C}{Q_1} \\\\ \etrip{P_2}{C}{Q_2} }
    { \etrip{P_1 \vee P_2}{C}{Q_1 \vee Q_2} }
\end{mathpar}


Note: We drop the equivalence rule, since it is immediately valid by our characterization of formulas by equivalence classes.

\subsection{Strongest postconditions}\label{app:el-sp}

\begin{lemma}\label{lem:el-sp-provable}
  $\etrip{P}{C}{\sp(P, C)}$
\end{lemma}
\begin{proof}
  By induction on $C$:

  \begin{description}
    \item[$\kskip$:]
      By \textsc{EX-Skip} $\etrip{P}{\kskip}{P}$ and $P \equiv \sp(P, \kskip)$ by definition.
    \item[$C \equiv x \coloneq E$:]~
      \begin{description}
        \item[(1)] $P \equiv \exists y (x = y \wedge P')$ for some $P'$ since $\exists y (x = y)$ is a tautology. WLOG assume $x \notin \fv(P')$ (instances of $x$ can be replaced with $y$).
        \item[(2)] $\etrip{x = y}{x \coloneq E}{x = E[x/y]}$ by \textsc{EX-Assign}.
        \item[(3)] $\etrip{x = y \wedge P'}{x \coloneq E}{x = E[x/y] \wedge P'}$ by \textsc{EX-Frame} and (2).
        \item[(4)] $\etrip{\exists y (x = y \wedge P')}{x \coloneq E}{\exists y (x = E[x/y] \wedge P')}$ by \textsc{EX-Exists} and (3).
        \item[(5)]
          \begin{align*}
            \sp(P, x \coloneq E) &\equiv \exists v (x = E[x/v] \wedge P[x/v]) &\text{defn} \\
              &\equiv \exists v (x = E[x/v] \wedge (\exists y (v = y \wedge P'))) &\text{(1)} \\
              &\equiv \exists y (x = E[x/y] \wedge P') &\text{logic}
          \end{align*}
        \item[(6)] $\etrip{P}{x \coloneq E}{\sp(P, x \coloneq E)}$ by (1), (4), and (5).
      \end{description}

    \item[$C_1; C_2$:]~
      \begin{description}
        \item[(1)] $\etrip{P}{C}{\sp(P, C_1)}$ by induction
        \item[(2)] $\etrip{\sp(P, C_1)}{C_2}{\sp(\sp(P, C_1), C_2)}$ by induction
        \item[(3)] $\etrip{P}{C_1; C_2}{\sp(\sp(P, C_1), C_2)}$ by \textsc{EX-Seq} using (1), (2)
        \item[(4)] $\sp(\sp(P, C_1), C_2) \equiv \sp(P, C_1; C_2)$ by definition
        \item[(5)] $\etrip{P}{C_1; C_2}{\sp(P, C_1; C_2)}$ by (3) and (4)
      \end{description}

    \item[$\ifstmt{E}{C_1}{C_2}$:]~
      \begin{description}
        \item[(1)] $\etrip{P \wedge E}{C_1}{\sp(P \wedge E, C_1)}$ by induction
        \item[(2)] $\etrip{P \wedge \neg E}{C_2}{\sp(P \wedge \neg E, C_2)}$ by induction
        \item[(3)] $\etrip{P \wedge E}{\ifstmt{E}{C_1}{C_2}}{\sp(P \wedge E, C_1)}$ by \textsc{EX-IfThen} using (1)
        \item[(4)] $\etrip{P \wedge \neg E}{\ifstmt{E}{C_1}{C_2}}{\sp(P \wedge \neg E, C_2)}$ by \textsc{EX-IfElse} using (2)
        \item[(5)] $\vdash ((P \wedge E) \vee (P \wedge \neg E)) ~\ifstmt{E}{C_1}{C_2} \\ (\sp(P \wedge E, C_1) \vee \sp(P \wedge \neg E, C_2))$ by \textsc{EX-Disj} using (3), (4)
        \item[(6)] $(P \wedge E) \vee (P \wedge \neg E) \equiv P$ by logic
        \item[(7)] $\sp(P \wedge E, C_1) \vee \sp(P \wedge \neg E, C_2) \equiv \\ \sp(P, \ifstmt{E}{C_1}{C_2})$ by definition
        \item[(8)] $\vdash (P) ~\ifstmt{E}{C_1}{C_2} \\ (\sp(P, \ifstmt{E}{C_1}{C_2}))$ by (5), (6), and (7) \qedhere
      \end{description}
  \end{description}
\end{proof}

\begin{lemma}\label{lem:el-sp-equiv}
  If $\etrip{P}{C}{Q}$ then $Q \equiv \sp(P, C)$.
\end{lemma}
\begin{proof}
  By induction on the derivation $\etrip{P}{C}{Q}$:
  \begin{description}
    \item[EX-Skip:]
      By inversion $C \equiv \kskip$ and $Q \equiv P$, and by definition $\sp(P, \kskip) \equiv P$. Thus $\sp(P, C) \equiv \sp(P, \kskip) \\ \equiv P \equiv Q$.

    \item[EX-Assign:]\leavevmode
      \begin{description}
        \item[(1)] $C \equiv x \coloneq E$ for some $x, E$ by inversion
        \item[(2)] $P \equiv x = E'$ for some $E'$ by inversion
        \item[(3)] $Q \equiv x = E[x/E']$ by inversion
        \item[(4)] $x \notin \fv(E')$ by inversion
        \item[(5)]
          \begin{align*}
            &\sp(P, C) \\
              &\quad\equiv \sp(x = E', x \coloneq E)  &\text{(1), (2)} \\
              &\quad\equiv \exists v (x = E[x/v] \wedge (x = E')[x/v]) &\text{defn $\sp$}\\
              &\quad\equiv \exists v (x = E[x/v] \wedge v = E') &\text{(4)} \\
              &\quad\equiv \exists v (x = E[x/E'] \wedge v = E') &\text{logic} \\
              &\quad\equiv x = E[x/E'] \wedge \exists v\, v = E' &\text{(4)} \\
              &\quad\equiv x = E[x/E'] &\text{logic} \\
              &\quad\equiv Q &\text{(3)}
          \end{align*}
      \end{description}

    \item[EX-IfThen:]\leavevmode
      \begin{description}
        \item[(1)] $C \equiv \ifstmt{E}{C_1}{C_2}$ by inversion
        \item[(2)] $P \equiv P' \wedge E$ for some $P'$ by inversion
        \item[(3)] $\etrip{P' \wedge E}{C_1}{Q}$ by inversion
        \item[(4)] $Q \equiv \sp(P' \wedge E, C_1)$ by induction using (3)
        \item[(5)]
          \begin{align*}
            &\sp(P, C) \\
            &\quad \equiv \sp(P' \wedge E, \ifstmt{E}{C_1}{C_2}) &\text{(1), (2)} \\
            &\quad \equiv \sp(P' \wedge E, C_1) &\text{\Cref{lem:sp-if}} \\
            &\quad \equiv Q &\text{(4)}
          \end{align*}
      \end{description}

      \item[EX-IfElse:]\leavevmode
        \begin{description}
          \item[(1)] $C \equiv \ifstmt{E}{C_1}{C_2}$ by inversion
          \item[(2)] $P \equiv P' \wedge \neg E$ for some $P'$ by inversion
          \item[(3)] $\etrip{P' \wedge \neg E}{C_2}{Q}$ by inversion
          \item[(4)] $Q \equiv \sp(P' \wedge \neg E, C_2)$ by induction using (3)
          \item[(5)]
            \begin{align*}
              &\sp(P, C) \\
              &\quad \equiv \sp(P' \wedge \neg E, \ifstmt{E}{C_1}{C_2}) &\text{(1), (2)} \\
              &\quad \equiv \sp(P' \wedge \neg E, C_2) &\text{\Cref{lem:sp-if}} \\
              &\quad \equiv Q &\text{(4)}
            \end{align*}
        \end{description}

      \item[EX-Seq:]\leavevmode
        \begin{description}
          \item[(1)] $C \equiv C_1; C_2$ for some $C_1, C_2$ by inversion
          \item[(2)] $\etrip{P}{C_1}{R}$ for some $R$ by inversion
          \item[(3)] $\etrip{R}{C_2}{Q}$ by inversion
          \item[(4)] $R \equiv \sp(P, C_1)$ by induction using (2)
          \item[(5)] $Q \equiv \sp(R, C_2)$ by induction using (3)
          \item[(6)]
            \begin{align*}
              \sp(P, C) &\equiv \sp(P, C_1; C_2) &\text{(1)} \\
                &\equiv \sp(\sp(P, C_1), C_2) &\text{defn $\sp$} \\
                &\equiv \sp(R, C_2) &\text{(4)} \\
                &\equiv Q &\text{(5)}
            \end{align*}
        \end{description}

      \item[EX-Exists:]\leavevmode
        \begin{description}
          \item[(1)] $P \equiv \exists x\, P'$ for some $P'$ by inversion
          \item[(2)] $Q \equiv \exists x\, Q'$ for some $Q'$ by inversion
          \item[(3)] $\etrip{P'}{C}{Q'}$ by inversion
          \item[(4)] $x \notin \fv(C)$ by inversion
          \item[(5)] $Q' \equiv \sp(P', C)$ by induction using (3)
          \item[(6)]
            \begin{align*}
              \sp(P, C) &\equiv \sp(\exists x\, P', C) &\text{(2)} \\
                &\equiv \exists x\, \sp(P', C) &\text{\Cref{lem:sp-exists}} \\
                &\equiv \exists x\, Q' &\text{(5)} \\
                &\equiv Q &\text{(2)}
            \end{align*}
        \end{description}

      \item[EX-Disj:]\leavevmode
        \begin{description}
          \item[(1)] $P \equiv P_1 \vee P_2$ for some $P_1, P_2$ by inversion
          \item[(2)] $Q \equiv Q_1 \vee Q_2$ for some $Q_1, Q_2$ by inversion
          \item[(3)] $\etrip{P_1}{C}{Q_1}$ by inversion
          \item[(4)] $\etrip{P_2}{C}{Q_2}$ by inversion
          \item[(5)] $Q_1 \equiv \sp(P_1, C)$ by induction using (3)
          \item[(6)] $Q_2 \equiv \sp(P_2, C)$ by induction using (4)
          \item[(7)]
            \begin{align*}
              \sp(P, C) &\equiv \sp(P_1 \vee P_2, C) &\text{(1)} \\
                &\equiv \sp(P_1, C) \vee \sp(P_2, C) &\text{\Cref{lem:sp-disj}} \\
                &\equiv Q_1 \vee Q_2 &\text{(5), (6)} \\
                &\equiv Q &\text{(2)}
            \end{align*}
        \end{description}

      \item[EX-Frame:]\leavevmode
        \begin{description}
          \item[(1)] $P \equiv P' \wedge R$ for some $P', R$ by inversion
          \item[(2)] $Q \equiv Q' \wedge R$ for some $Q'$ by inversion
          \item[(3)] $\mod(C) \cap \fv(R) = \emptyset$ by inversion
          \item[(4)] $\etrip{P'}{C}{Q'}$ by inversion
          \item[(5)] $\sp(P', C) \equiv Q$ by induction using (4)
          \item[(6)] $\sp(P' \wedge R, C) \equiv R \wedge \sp(P', C)$ by \cref{lem:sp-frame} using (3)
          \item[(7)]
            \begin{align*}
              \sp(P, C) &\equiv \sp(P' \wedge R, C) &\text{(1)} \\
                &\equiv R \wedge \sp(P', C) &\text{(6)} \\
                &\equiv R \wedge Q' &\text{(5)} \\
                &\equiv Q &\text{(2)} &\qedhere
            \end{align*}
        \end{description}
  \end{description}
\end{proof}

\begin{theorem}\label{thm:el-iff-sp}
  $\etrip{P}{C}{Q} \iff Q \equiv \sp(P, C)$
\end{theorem}
\begin{proof}
  $\Longrightarrow$: \Cref{lem:el-sp-equiv}; 
  $\Longleftarrow$: \Cref{lem:el-sp-provable}
\end{proof}

\subsection{Gradual exact logic}\label{app:el-gradual}

Valid triples in gradual exact logic are denoted $\getrip{\grad{P}}{C}{\grad{Q}}$.

\begin{definition}
  The \emph{concretization} $\gamma : \GFormula \to \\ \mathcal{P}(\Formula)$ maps a gradual formula to the set of all formulas it can represent:
  \begin{align*}
    \gamma(P) &\coloneq \{ P \} \\
    \gamma(\ttt{?} \wedge P) &\coloneq \{ P' \in \SatFormula \mid P' \Rightarrow P \}
  \end{align*}
\end{definition}

\begin{definition}
  Deductions in gradual exact logic directly are lifted deductions in exact logic:
  \begin{align*}
    &\getrip{\grad{P}}{C}{\grad{Q}} \iffdef \etrip{P}{C}{Q} \\
    &\hspace{2em} \text{for some $P \in \gamma(\grad{P})$ and $Q \in \gamma(\grad{Q})$}
  \end{align*}
\end{definition}

\begin{theorem}\label{thm:hl-as-gel}
  For $P \in \SatFormula$,
  $$\getrip{P}{C}{\ttt{?} \wedge Q} \iff \htrip{P}{C}{Q}.$$
  That is, except for the vacuous case where $P \equiv \bot$, gradualizing the postcondition exactly characterizes deductions in Hoare logic.
\end{theorem}
\begin{proof}
  $\Longrightarrow$:
    \begin{description}
      \item[(1)] $\getrip{P}{C}{\ttt{?} \wedge Q}$ by assumptino
      \item[(2)] $\etrip{P}{C}{Q'}$ for some $Q' \in \gamma(\ttt{?} \wedge Q)$ by (1)
      \item[(3)] $\sp(P, C) \equiv Q'$ by \cref{lem:el-sp-equiv} using (2)
      \item[(4)] $Q' \Rightarrow Q$ by definition of $\gamma$ and (2)
      \item[(5)] $\sp(P, C) \Rightarrow Q$ by (3) and (4)
      \item[(6)] $\htrip{P}{C}{Q}$ by \cref{lem:sp-imply-hoare} using (5)
    \end{description}

  $\Longleftarrow$:
  \begin{description}
    \item[(1)] $\htrip{P}{C}{Q}$ by assumption
    \item[(2)] $P \in \SatFormula$ by assumption
    \item[(3)] $\sp(P, C) \Rightarrow Q$ by \cref{lem:hoare-imply-sp} using (1)
    \item[(4)] $\sp(P, C) \in \SatFormula$ by \cref{lem:sp-sat} using (2)
    \item[(5)] $\sp(P, C) \in \gamma(\ttt{?} \wedge Q)$ by definition of $\gamma$ using (3) and (4)
    \item[(6)] $\etrip{P}{C}{\sp(P, C)}$ by \cref{lem:el-sp-provable}
    \item[(7)] $\getrip{P}{C}{\ttt{?} \wedge Q}$ by definition using (5) and (6) \qedhere
  \end{description}
\end{proof}

\begin{theorem}\label{thm:il-as-gel}
  If $Q \in \SatFormula$,
  $$\getrip{\ttt{?} \wedge P}{C}{Q} \iff \itrip{P}{C}{Q}$$
  That is, except in the vacuous case where $Q \equiv \bot$, gradualizing the precondition exactly characterizes deductions in incorrectness logic.
\end{theorem}
\begin{proof}
  $\Longrightarrow$:
  \begin{description}
    \item[(1)] $\getrip{\ttt{?} \wedge P}{C}{Q}$ by assumption
    \item[(2)] $\etrip{P'}{C}{Q}$ for some $P' \in \gamma(\ttt{?} \wedge P)$ by definition using (1)
    \item[(3)] $P' \Rightarrow P$ by definition of $\gamma$ using (2)
    \item[(4)] $Q \equiv \sp(P', C)$ by \cref{lem:el-sp-equiv} using (2)
    \item[(5)] $\sp(P', C) \Rightarrow \sp(P, C)$ by \cref{lem:sp-strengthen} using (3)
    \item[(6)] $Q \Rightarrow \sp(P, C)$ by (4) and (5)
    \item[(7)] $\itrip{P}{C}{Q}$ by \cref{thm:hoare-iff-sp} using (6)
  \end{description}

  $\Longleftarrow$:
  \begin{description}
    \item[(1)] $\itrip{P}{C}{Q}$ by assumption
    \item[(2)] $Q \in \SatFormula$ by assumption
    \item[(3)] $P \wedge \wp(C, Q) \in \SatFormula$ by \cref{lem:il-wp-sat} using (1) and (2)
    \item[(4)] $P \wedge \wp(C, Q) \Rightarrow P$ by logic
    \item[(5)] $P \wedge \wp(C, Q) \in \gamma(\ttt{?} \wedge P)$ by definition of $\gamma$ using (3) and (4)
    \item[(6)] $Q \Rightarrow \sp(P, C)$ by \cref{thm:hoare-iff-sp} using (1)
    \item[(7)] $\sp(P \wedge \wp(C, Q), C) \equiv Q$ by \cref{lem:sp-wp-fixpoint-p} using (6)
    \item[(8)] $\etrip{P \wedge \wp(C, Q)}{C}{Q}$ by \cref{thm:el-iff-sp} using (7)
    \item[(9)] $\getrip{\ttt{?} \wedge P}{C}{Q}$ by (5) and (8) \qedhere
  \end{description}
\end{proof}

\end{document}
\endinput